\documentclass[runningheads]{llncs}
\usepackage[T1]{fontenc}
\usepackage{graphicx}
\usepackage{tikz}
\usepackage{amsmath}
\usepackage{amssymb}

\usepackage{verbatim}
\usepackage[linesnumbered]{algorithm2e}
\usepackage[caption=false]{subfig}
\usepackage{url}
\usepackage{tikz-qtree}
\usepackage{wrapfig}

\newtheorem{observation}{Observation}

\begin{document}
\title{Deciding whether an Attributed Translation can be realized by a Top-Down Transducer} 
\titlerunning{Deciding whether an Attributed Translation can be realized by Transducers}
%
\author{Sebastian Maneth \and
	Martin Vu}
\authorrunning{S. Maneth and
	M. Vu}
%
\institute{Universit\"at Bremen, Germany\\
	\email{\{maneth,martin.vu\}@uni-bremen.de} }

\maketitle

	\begin{abstract}
	We prove that for a given partial deterministic attributed tree transducer
	with monadic output, it is decidable
	whether or not an equivalent top-down tree transducer (with or without look-ahead) exists.
	We present a procedure that constructs such an equivalent top-down tree transducer (with or without look-ahead)
	if it exists.
	We then show that our results can be extended to arbitrary nondeterministic attributed tree transducer
	with look-around that have monadic output.
\end{abstract}


\section{Introduction}
First invented in the 1970's in
the context of compilers and mathematical linguistics,
tree transducers are fundamental devices with far ranging applications
including 
picture generation~\cite{DBLP:series/txtcs/Drewes06},
 network intrusion detection~\cite{Sugiyama2014}, 
security~\cite{DBLP:journals/iandc/KustersW07},
and XML databases~\cite{DBLP:conf/icde/HakutaMNI14}.

Two prominent types of tree transducer are the
\emph{top-down tree 
transducer}~\cite{DBLP:journals/mst/Rounds70,DBLP:journals/jcss/Thatcher70}
and the \emph{attributed tree transducer} \cite{DBLP:journals/actaC/Fulop81}.
As its name implies, a top-down tree transducer processes its input tree strictly in a top-down fashion
meaning that its states are only moving `downwards' in the input tree.
In contrast, the attributes of an attributed tree transducer can move `downwards' as well as  `upwards'
when processing an input tree. 
In the context of strings, one possible  pair of respective counterparts of
top-down tree transducers and attributed tree transducer are  one-way  transducers and 
two-way  transducers. Mirroring the behavior of their tree transducer counterparts,
when processing an input string,  
states of a one-way transducer are limited to moving strictly from left to right,
while states of a two-way transducer can move from left to right and from right to left.
Unsurprisingly, 
attributed tree transducers and two-way  transducer are therefore strictly more expressive
than top-down tree transducer and one-way string transducer, respectively.
However this expressiveness comes at the cost of complexity; generally speaking
they are much more complex devices than
a top-down tree transducer and one-way  transducers, respectively.
Hence, for either an attributed tree transducer or a two-way  transducer, it is 
a natural question to ask:
Can its translation also be realized by the respective simpler device?
And if so, can we construct the simpler device?

For two-way transducers, this question has been answered in~\cite{DBLP:conf/lics/FiliotGRS13}.
Specifically, in~\cite{DBLP:conf/lics/FiliotGRS13}, a procedure is introduced that given a 
two-way transducer, decides whether or not its translation can also be realized
by a one-way transducer and in the affirmative case constructs such a one-way transducer.
However, it is also shown the given procedure has non-elementary complexity.
In a subsequent paper~\cite{DBLP:journals/lmcs/BaschenisGMP18}, the decision procedure
has been improved upon and simplified leading to a triply exponential time complexity.
Given the result in~\cite{DBLP:conf/lics/FiliotGRS13}, one wonders whether the same result can be 
obtained for  attributed tree transducers and
top-down tree transducers, i.e., given an   attributed tree transducer is it decidable
whether or not an equivalent top-down tree transducer exists? And if so, can we construct it?

In general, decision procedures of this kind offer several advantages; as in our case the smaller class
may be less complex and thus
may be more efficient to evaluate (i.e., it may use less resources). Other possible benefits are that 
the smaller class may enjoy better closure properties than the larger class.

In this paper, we address this problem for a subclass of attributed tree transducers.
In particular, we consider attributed tree transducer with \emph{monadic output}
meaning that all nodes of output trees produced by the transducer have at most one child node
making the output trees essentially ``strings''.
Initially, we show that it is decidable whether or not
for a given deterministic attributed tree transducer $A$ with monadic output an equivalent deterministic top-down transducer $T$ with
look-ahead 
exists by reducing the problem to the question of whether or not a given two-way transducer can be
defined by a one-way transducer,
before extending our results to more complex types of attributed tree transducers.
To show that the decision problem  has a solution for such attributed tree transducer,
 we  first test whether $A$ has the \emph{single-path property}.
 The latter essentially 
 means that $A$
can be equipped with `look-ahead' so that $A$ only processes a single input path of an input tree. 
A look-ahead is a deterministic bottom-up relabeling which preprocesses input trees for $A$.
Intuitively, 
$A$ only processes a single input path of an input tree $t$ if all nodes of $t$ that
attributes of $A$ visit 
occur in a node sequence 
$v_1,\dots,v_n$ where $v_i$ is the parent of $v_{i+1}$.
This property is derived from the fact that any
top-down tree transducer $T$ with look-ahead  that is equivalent to $A$ processes its input tree in exactly the same fashion. 
In particular, being equivalent to $A$ means that $T$ also generates monadic output trees and for any
top-down tree transducer with look-ahead that only generates monadic output trees,
it holds that its states only process nodes occurring on a single input path.
The idea is that if a single input path is sufficient for a top-down tree transducer $T$ (with look-ahead)
to generate its output tree then it should be sufficient for $A$ (equipped with look-ahead) as~well.
Assume that $A$ has the single-path property.
We then show that $A$ can be converted into a two-way transducer $T_W$.
Given $T_W$
we  apply the procedure of~\cite{DBLP:journals/lmcs/BaschenisGMP18} checking whether or not a one-way-transducer equivalent to $T_W$
exists.
It can be shown that the procedure of~\cite{DBLP:journals/lmcs/BaschenisGMP18} yields a one-way transducer equivalent to $T_W$
if and only if a top-down tree transducer $T$ with look-ahead equivalent to $A$ exists.
We show that after computing a one-way transducer $T_O$  equivalent to $T_W$ using the procedure~of~\cite{DBLP:journals/lmcs/BaschenisGMP18},
we can construct a top-down tree transducer  with look-ahead equivalent to $A$ from $T_O$.

Extending the result above, we show that even for
nondeterministic  attributed tree transducers $\breve{A}$ with `look-around' and monadic output,
it is decidable whether or not an equivalent top-down transducer $\breve{T}$ with look-ahead exists.
Look-around is a relabeling device similar to but more expressive than look-ahead
which was introduced by Bloem and Engelfriet~\cite{DBLP:journals/jcss/BloemE00}
due to its better closure properties.
To extend our result to such  transducers, we  show that (a)
for an attributed tree transducers
with look-around and monadic output, it is decidable whether or not it is \emph{functional}, i.e,
whether or not its translation is a function, (b)
functional and  deterministic attributed tree transducer with look-around and monadic output
describe the same class of translations and
(c)
for deterministic attributed tree transducers with look-around and monadic output,
it is decidable whether or not an equivalent top-down transducer with look-ahead exists.
Finally, we show that due to the result of~\cite{DBLP:conf/icalp/ManethS20}, it is decidable in which cases the 
look-ahead can be removed from $\breve{T}$ as well. 

We remark that due to Proposition~9.3 in~\cite{DBLP:journals/lmcs/BaschenisGMP18}, 
deciding for a non-functional attributed tree transducer with monadic output whether or not an equivalent non-functional
top-down tree transducer exists is undecidable.
Furthermore note that
nondeterministic functional top-down tree transducer with look-ahead and deterministic top-down transducer 
with look-ahead define the same class of translations~\cite{DBLP:journals/ipl/Engelfriet78}.
Therefore, confining ourselves to deterministic top-down transducers 
instead of functional ones in this paper is not a restriction.

Note that in the presence of origin, it is well known that even for (nondeterministic)
macro tree transducers (which are strictly more expressive than attributed tree transducers)
it is decidable whether or not an origin-equivalent deterministic top-down tree transducer with
look-ahead exists~\cite{DBLP:journals/iandc/FiliotMRT18}. 
Informally, the presence of origin means that 
the semantic of a transducer allow us to trace
for each node of an output tree the unique node of the input tree that created it.
In the absence of origin, the only  definability results for attributed transducers that
we are aware of, is that it is decidable for such transducers (and even for macro tree transducers)
whether or not they are 
(1)~of linear size increase~\cite{DBLP:journals/siamcomp/EngelfrietM03}
(and if so an equivalent single-use restricted attributed tree transducer can be constructed;
see~\cite{DBLP:journals/iandc/EngelfrietM99})
or
(2)~of either  linear height-increase
or linear size-to-height increase~\cite{gallot2024deciding}.

A preliminary version of 
this paper has been presented at the
27th International Conference on
Implementation and Application of Automata 
(CIAA) 2023 \cite{DBLP:conf/wia/ManethV23}.
\section{Preliminaries}\label{preliminaries}
Denote by $\mathbb{N}$ the set of natural numbers.
For $k\in\mathbb{N}$, we denote by $[k]$ the set $\{1,\dots,k\}$.
A set $\Sigma$ is \emph{ranked} if  
each symbol of the set is associated with a \emph{rank}, that is, a non-negative integer.
We write $\sigma^{k}$ to denote that the symbol $\sigma$ has
rank~$k$.
By $\Sigma_k$ we denote the set of  all symbols of $\Sigma$ which have rank $k$.
We require that for $k'\neq k$, $\Sigma_{k'}$ and $\Sigma_k$ are disjoint.
If $\Sigma$ is finite then we also call $\Sigma$ a \emph{ranked alphabet}. 

The set $T_\Sigma$ of 
\textit{trees over $\Sigma$} is defined as the smallest set of strings such that
if $\sigma\in \Sigma_k$, $k\geq 0$, and $t_1,\dots,t_k \in T_\Sigma$ then
$\sigma(t_1,\dots, t_k)$ is in $T_\Sigma$. 
For $k=0$, we simply write $\sigma$
instead of $\sigma()$.
The nodes of a tree $t\in T_\Sigma$ are referred to by strings over~$\mathbb{N}$.
In particular, for $t=\sigma (t_1,\dots,t_k)$, we define $V(t)$, the set of nodes of $t$, 
as $V(t)=\{\epsilon\}\cup \{iu\mid i\in [k] \text{ and } u\in V(t_i)\}$,
where $\epsilon$ is the \emph{empty string}.
For better readability, we add dots between numbers,
 e.g. for the tree $t=f(a,f(a,b))$ we have $V(t)=\{\epsilon,1,2,2.1,2.2\}$.
For a node  $v\in V(t)$, $t[v]$ denotes the label of $v$, 
$t/v$ is the subtree of $t$ rooted at $v$, and
$t[v \leftarrow t']$ is  obtained from $t$ by replacing $t/v$ by $t'$.
For instance, we have $t[1]=a$, $t/2= f(a,b)$ and $t[1\leftarrow b]=f(b,f(a,b))$ for $t=f(a,f(a,b))$.
The node $v$ is called a (proper) ancestor of a node $v'\in V(t)$ if $v$ is a (proper) prefix of $v'$.
For a tree $s$ denote by $|s|:= |V(s)|$ the \emph{size} of $s$.

For a set $\Lambda$ disjoint with $\Sigma$, we define $T_\Sigma [\Lambda]$ as $T_{\Sigma'}$ where $\Sigma'_0 =\Sigma_0\cup \Lambda$ and $\Sigma_k'=\Sigma_k$ for $k>0$.
We call  a tree $t'\in T_{\Sigma} [\Lambda]$  a \emph{prefix} of a tree  $t\in T_{\Sigma}$
if $t$ can be obtained from $t'$  by replacing nodes labeled by symbols in $\Lambda$ 
by  trees over~$\Sigma$, i.e., if 
for $V=\{v\in V(t') \mid t'[v]\in \Lambda\}$
a set $\{t_v \in T_\Sigma \mid v\in V\}$ exists such that
$t=t'[v\leftarrow t_v \mid v\in V]$.

\subsection{Attributed Tree Transducers}
	A \emph{(partial nondeterministic) attributed tree transducer} (or  $att$ for short) is a tuple $A=(S,I,\Sigma,\Delta,a_0,R)$ where
	\begin{itemize}
		\item $S$ and $I$ are disjoint finite sets of
		\emph{synthesized attributes} and \emph{inherited attributes}, respectively
		\item $\Sigma$ and $\Delta$ are ranked alphabets of \emph{input} and \emph{output symbols}, respectively
		\item  $a_0\in S$ is  the \emph{initial attribute} and
		\item $R=(R_\sigma\mid \sigma \in \Sigma\cup \{\#\})$ is a collection of finite sets of rules.
	\end{itemize}
	We implicitly assume $atts$ to include a unique symbol $\#\notin \Sigma$ of rank $1$, the so-called
	\emph{root marker}, that may only occur at the root of input trees.

	In the following, we define the rules of an $att$.
	Let $\sigma\in \Sigma$ be of rank $k\geq 0$. 
	Furthermore, let $\pi$ be a variable for nodes. Then the set $R_\sigma$ contains
	\begin{itemize}
		\item 	arbitrarily many rules of the form $a(\pi)\rightarrow \xi$ for every $a\in S$ and
		\item  arbitrarily many rules of the form $b(\pi i)\rightarrow \xi'$ for every $b\in I$ and $i\in [k]$,
	\end{itemize}
 where
	$\xi,\xi'\in T_\Delta [\{a'(\pi i) \mid a'\in S, i\in [k]\} \cup \{b'(\pi) \mid b'\in I\} ]$.
	We define the set $R_\#$ analogously with the restriction that $R_\#$ contains \emph{no}
	rules with synthesized attributes  on the left-hand side.
	Replacing  `arbitrarily many rules' by `at most one rule' in the definition of the
	rule sets of $R$, we obtain the notion of \emph{(partial) deterministic $att$ (or $datt$)}.
	For the $att$ $A$ and the attribute $a\in S$, we denote by $\text{RHS}_A (\sigma, a(\pi))$
	the set of all right-hand sides of rules in $R_\sigma$ that are of the form
	$a(\pi)\rightarrow \xi$.
	For $b\in I$, the sets $\text{RHS}_A (\sigma, b(\pi i))$ with $i\in [k]$ and $\text{RHS}_A (\#, b(\pi 1))$ are defined analogously.	
	
If $I=\emptyset$ then we call $A$ a \emph{top-down tree transducer} and $S$ a set of \emph{states} instead of attributes.
Additionally, if $A$ is also deterministic then we call $A$ a \emph{deterministic top-down tree transducer} (or simply a $dt$).
For a top-down tree transducer and a symbol $\sigma$ of rank $k\geq 0$, we commonly write
$q(\sigma (x_1,\dots,x_k)) \rightarrow t'$
instead of
$q(\pi) \rightarrow t \in R_\sigma$,
where $t'$ is obtained from $t$
by replacing occurrences of $\pi i$, $i\in [k]$, by $x_i$, e.g.,
for $t= f(q_1 (\pi 1), q_2 (\pi 2) )$ we have $t'=f(q_1 (x_1), q_2 (x_2))$.

We say that $A$ is an $att$ with \emph{monadic output}, if
all output symbols of $A$ are at most of rank $1$.\\

\textbf{Attributed Tree Translation.}\quad
We now define the semantics of $A$. Denote by $T_{\Sigma^\#}$ the set $\{\# (s) \mid s\in T_\Sigma\}$.
For a tree $s\in T_{\Sigma}\cup T_{\Sigma^\#}$, we define
$
\text{SI}(s)= \{\alpha (v) \mid \alpha\in S\cup I, v\in V(s)\}.
$
Furthermore, we define that for the node variable $\pi$,
$\pi 0 = \pi$  and that for a node $v$,
$v.0=v$.
Let $t,t' \in T_\Delta [\text{SI}(s)]$. 
We write $t\Rightarrow_{A,s} t'$ if 
$t'$ is obtained from $t$
by substituting  a leaf of $t$ labeled by $\gamma (v.i)$, with $i=0$ if $\gamma\in S$ and $i>0$ if $\gamma\in I$,
by 
$\xi [\pi\leftarrow v]$,
where $\xi\in \text{RHS}_A (s[v], \gamma (\pi i))$ and 
$[\pi\leftarrow v]$ denotes the substitution that replaces all occurrences
of $\pi$ by the node $v$. For instance, for $\xi_1=f(b(\pi))$ and $\xi_2=f(a(\pi 2))$ where $f$ is a symbol of
rank~$1$, $a\in S$
and $b\in I$, we have $\xi_1[\pi\leftarrow v]=f(b(v))$ and $\xi_2[\pi\leftarrow v]=f(a(v.2))$.
As usual, denote by 
$\Rightarrow_{A,s}^+$ and
$\Rightarrow_{A,s}^*$ the transitive closure and the reflexive-transitive closure of $\Rightarrow_{A,s}$, respectively.

The 
\emph{translation realized by $A$}, denoted by $\tau_A$, is the set
\[
\{(s,t) \in T_\Sigma \times T_\Delta\mid  a_0(1)\Rightarrow_{A,s^\#}^* t\},
\]
where subsequently $s^\#$ denotes the tree $\# (s)$.
If $\tau_A$ is a  partial function then we say that $A$ is a \emph{functional} $att$.
Furthermore, if $\tau_A$ is a partial function 
then we also write $\tau_A (s)=t$ if $(s,t)\in \tau_A$ and say that on input $s$, $A$ produces the 
tree~$t$. 
Denote by $\text{dom} (A)$
the \emph{domain of $A$}, i.e., the set of all $s\in T_\Sigma$ such that
$(s,t)\in \tau_A$ for some $t\in T_\Delta$.
Similarly, $\text{range} (A)$ denotes the \emph{range of $A$}, i.e.,
the set of all $t\in T_\Delta$ such that
for some $s\in T_\Sigma$, $(s,t)\in \tau_A$. 

\begin{example}\label{att example}
	Consider the $att$ $A_1=(S,I,\Sigma,\Delta,a,R)$ where 
	$\Sigma= \{f^2, e^0\}$ and $\Delta=\{ g^1, e^0\}$.
	Let the set of attributes of $A_1$ be given by $S=\{ a\}$ and $I=\{b\}$.
	We define 
	\[
	R_f=\{a(\pi)  \rightarrow   a (\pi 1),\ b(\pi 1) \rightarrow  a (\pi 2),\
	 b(\pi 2) \rightarrow  b(\pi)\}.
	 \]
	Furthermore, we define 
	\[
	R_{\#}=\{b(\pi 1) \rightarrow e\}
	\text{ and }
	R_e=\{a(\pi) \rightarrow g(b(\pi))\}.
	\]
	The tree transformation realized by $A_1$ 
	contains all pairs $(s,t)$ such that if $s$ has $n$ leaves, then $t$ is the tree over~$\Delta$ that
	contains $n$ occurrences of the symbol $g$. 
	For instance on input $s=f(f(e,e), f(e,e))$, $A_1$ outputs a tree with
	four occurrences of $g$. The corresponding translation is shown in Figure~\ref{fig0}.
	\begin{figure}[h!]
		\vspace{-0.3cm}
		\centering
		\begin{tikzpicture}
	\draw (2,0) node {$a(1)$};
	\draw (3,0) node {$\Rightarrow$};
	\draw (4,0) node {$a(1.1)$};
	\draw (5,0) node {$\Rightarrow$};
	\draw (6,0) node {$a(1.1.1)$};

	\draw (7,0) node {$\Rightarrow$};
	\draw (8,0.25) node {$g$};
	\draw (8,-0.1) -- (8,0.1);
	\draw (8,-0.3) node {$b(1.1.1)$};

	\draw (9,0) node {$\Rightarrow$};
	\draw (10,0.25) node {$g$};
	\draw (10,-0.1) -- (10,0.1);
	\draw (10,-0.3) node {$a(1.1.2)$};
	
	\draw (1,-1.2) node {$\Rightarrow$};
	\draw (2,-0.9) node {$g$};
	\draw (2,-1.05) -- (2,-1.25);
	\draw (2,-1.4) node {$g$};
	\draw (2,-1.55) -- (2,-1.75);
	\draw (2,-1.9) node {$b(1.1.2)$};

	\draw (3,-1.2) node {$\Rightarrow$};
	\draw (4,-0.9) node {$g$};
	\draw (4,-1.05) -- (4,-1.25);
	\draw (4,-1.4) node {$g$};
	\draw (4,-1.55) -- (4,-1.75);
	\draw (4,-1.9) node {$b(1.1)$};
	
	\draw (5,-1.2) node {$\Rightarrow$};
	\draw (6,-0.9) node {$g$};
	\draw (6,-1.05) -- (6,-1.25);
	\draw (6,-1.4) node {$g$};
	\draw (6,-1.55) -- (6,-1.75);
	\draw (6,-1.9) node {$a(1.2)$};

	\draw (7,-1.2) node {$\Rightarrow$};
	\draw (8,-0.9) node {$g$};
	\draw (8,-1.05) -- (8,-1.25);
	\draw (8,-1.4) node {$g$};
	\draw (8,-1.55) -- (8,-1.75);
	\draw (8,-1.9) node {$a(1.2.1)$};

	\draw (9,-1.2) node {$\Rightarrow$};
	\draw (10,-0.9) node {$g$};
	\draw (10,-1.05) -- (10,-1.25);
	\draw (10,-1.4) node {$g$};
    \draw (10,-1.55) -- (10,-1.75);
    \draw (10,-1.9) node {$g$};
    \draw (10,-2.05) -- (10,-2.25);
    \draw (10,-2.4) node {$b(1.2.1)$};
    
    \draw (1,-3.8) node {$\Rightarrow$};
    \draw (2,-3.3) node {$g$};
    \draw (2,-3.45) -- (2,-3.65);
    \draw (2,-3.8) node {$g$};
    \draw (2,-3.95) -- (2,-4.15);
    \draw (2,-4.3) node {$g$};
    \draw (2,-4.45) -- (2,-4.65);
    \draw (2,-4.8) node {$a(1.2.2)$};
    
    \draw (3,-3.8) node {$\Rightarrow$};
    \draw (4,-3) node {$g$};
    \draw (4,-3.15) -- (4,-3.35);
    \draw (4,-3.5) node {$g$};
    \draw (4,-3.65) -- (4,-3.85);
    \draw (4,-4) node {$g$};
    \draw (4,-4.15) -- (4,-4.35);
    \draw (4,-4.5) node {$g$};
    \draw (4,-4.65) -- (4,-4.85);
    \draw (4,-5) node {$b(1.2.2)$};
    
    \draw (5,-3.8) node {$\Rightarrow$};
\draw (6,-3) node {$g$};
\draw (6,-3.15) -- (6,-3.35);
\draw (6,-3.5) node {$g$};
\draw (6,-3.65) -- (6,-3.85);
\draw (6,-4) node {$g$};
\draw (6,-4.15) -- (6,-4.35);
\draw (6,-4.5) node {$g$};
\draw (6,-4.65) -- (6,-4.85);
\draw (6,-5) node {$b(1.2)$};

\draw (7,-3.8) node {$\Rightarrow$};
\draw (8,-3) node {$g$};
\draw (8,-3.15) -- (8,-3.35);
\draw (8,-3.5) node {$g$};
\draw (8,-3.65) -- (8,-3.85);
\draw (8,-4) node {$g$};
\draw (8,-4.15) -- (8,-4.35);
\draw (8,-4.5) node {$g$};
\draw (8,-4.65) -- (8,-4.85);
\draw (8,-5) node {$b(1)$};
    
\draw (9,-3.8) node {$\Rightarrow$};
\draw (10,-3) node {$g$};
\draw (10,-3.15) -- (10,-3.35);
\draw (10,-3.5) node {$g$};
\draw (10,-3.65) -- (10,-3.85);
\draw (10,-4) node {$g$};
\draw (10,-4.15) -- (10,-4.35);
\draw (10,-4.5) node {$g$};
\draw (10,-4.65) -- (10,-4.85);
\draw (10,-5) node {$e$};
		\end{tikzpicture}
			\caption{Translation of the $att$ $A_1$ in Example~\ref{att example} on input $s=f(f(e,e), f(e,e))$.}
	\label{fig0}
	\end{figure}
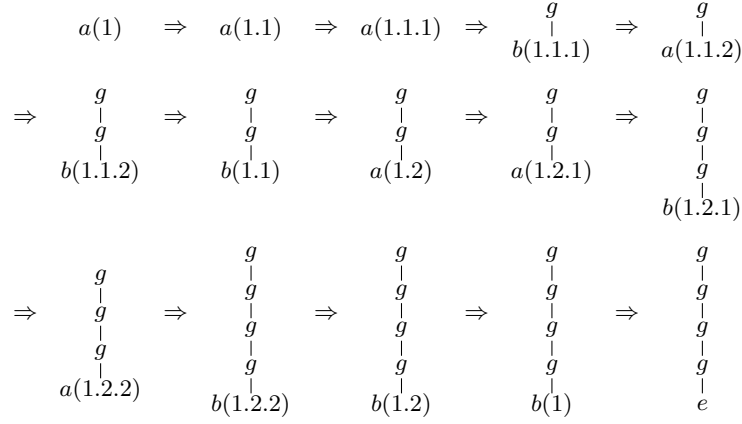
	Note that for better readability we have simply written $\Rightarrow$ instead of $\Rightarrow_{A_1,s^\#}$.
	We remark that the domain of $A_1$ is $T_\Sigma$ and its range is $T_\Delta\setminus\{e\}$. 
\end{example}

We emphasize that we always consider input trees to be trees over $\Sigma$. 
The root marker is only a technical requirement. For instance, without the root marker, the translation of the $att$ $A_1$ in
Example~\ref{att example} cannot be realized by an $att$.\\

\textbf{Circularity and Is-Dependency.}\quad
Note that by definition $atts$ are allowed to be \emph{circular}. 
We say that an $att$ $A$ is \emph{circular} if  $s\in T_{\Sigma}$, $\alpha (v) \in \text{SI} (s^\#)$ and $t\in T_\Delta [\text{SI} (s^\#)]$  exists such that
$\alpha (v) \Rightarrow_{A, s^\#}^+ t$ and $\alpha (v)$ occurs in~$t$. 
It is well known that circularity is a decidable property~\cite{DBLP:journals/mst/Knuth68}.
To test whether or not $A$ is circular, we compute the set of all \emph{is-dependencies} of~$A$,
i.e., the set $\text{ISD}_A=\{\text{ISD}_A (s) \mid s\in T_\Sigma\}$, where for a tree $s$,
we define the is-dependency of $s$ as
\[
\text{ISD}_A (s)= \{ (b,a)\in I\times S\mid  \exists\ t'\in T_\Delta [\text{SI} (s)]: a(\epsilon) \Rightarrow_{A, s}^* t'  \text{ and } b(\epsilon)
\text{ occurs in } t'\}.
\]
Note that $\text{ISD}_A (s)$ can be computed inductively in a bottom-up fashion, i.e.,
if for $s=\sigma (s_1,\dots,s_k)$, the is-dependencies of $s_1,\dots,s_k$ are known, then the is-dependency of
$s$ can be easily computed using the rules in~$R_\sigma$.

By definition $\text{ISD}_A$ is finite.
Furthermore, $\text{ISD}_A$ effectively computable.
For $\sigma\in \Sigma\cup\{\#\}$ of rank $k>1$ and  $\text{is}_1, \dots,\text{is}_k\in \text{ISD}_A$, we define 
the directed graph $G^A_{\sigma,\text{is}_1, \dots,\text{is}_k} =(V,E)$
where $V=\{\alpha (j) \mid \alpha \in S\cup I,\ j\in [k]\}$ and where for $a\in S$, $b\in I$ and $i,j\in [k]$
\begin{itemize}
	\item $(a( i), b(j)) \in E$ if $a(\pi i)$ occurs in some $t\in \text{RHS}_A (\sigma, b(\pi j))$
	\item  $(b(j), a(j)) \in E$ if $(b,a)\in \text{is}_j$.
\end{itemize}
It holds that $A$ is circular if and only if  $\sigma\in \Sigma\cup\{\#\}$ of rank $k>1$ and  $\text{is}_1, \dots,\text{is}_k\in \text{ISD}_A$ exist such that $G^A_{\sigma,\text{is}_1, \dots,\text{is}_k}$
has a cycle. 
Hence, the circularity of $A$ is decidable~\cite{DBLP:journals/mst/Knuth68}.\\

\textbf{Look-Ahead.}\quad
Before we  define \emph{attributed tree transducer with look-ahead}, we define
\emph{bottom-up relabelings}.
Formally, a \emph{bottom-up relabeling} $B$ is a tuple $(P,\Sigma,\Sigma',F,R)$
where $P$ is the set of states, $\Sigma$, $\Sigma'$ are finite ranked alphabets
and $F\subseteq P$ is the set of final states. 
For $\sigma\in \Sigma$ and $p_1,\dots,p_k\in P$,
the set $R$ contains 
at most one rule of the form
$
\sigma (p_1 (x_1),\dots, p_k (x_k))\rightarrow p (\sigma' (x_1,\dots,x_k))
$
where $p\in P$ and $\sigma'\in \Sigma'$.
The rules of $B$ induce a derivation relation $\Rightarrow_B^*$
which is defined inductively as follows: 
\begin{itemize}
	\item Let $\sigma\in \Sigma_0$ and $\sigma \rightarrow p (\sigma')$ be a rule in $R$.
	Then $\sigma \Rightarrow_B  p (\sigma')$.
	\item Let $s=\sigma (s_1,\dots, s_k)$ with $\sigma\in \Sigma_k$, $k>0$,
	and $s_1,\dots, s_k \in T_\Sigma$. 
	For $i\in [k]$, let $s_i \Rightarrow_B^* p_i (s_i')$. 
	Furthermore, let $\sigma (p_1 (x_1),\dots, p_k (x_k))\rightarrow p (\sigma' (x_1,\dots,x_k))$ be a rule in
	$R$. Then $s\Rightarrow_B^* p (\sigma' (s_1',\dots, s_k') )$.
\end{itemize}
For $s\in T_\Sigma$ and $p\in P$,
we write $s\in \text{dom}_B (p)$ if $s\Rightarrow_B^* p (s')$ for some tree $s'\in T_{\Sigma'}$.
The translation realized by $B$ is given by
$
\tau_B= \{(s,s') \in T_\Sigma\times T_{\Sigma'} \mid s\Rightarrow_B^* p(s') \text{ where } p\in F\}.
$
Since $\tau_B$ is a partial function, we also write $\tau_B (s) =t$ if~$(s,t)\in \tau_B$.
The domain and the range of $B$ are defined in the obvious way. 

We define an \emph{attributed tree transducer with look-ahead (or $att^R$)}
as a pair
$\hat{A}=(B,A)$ where $B$ is a
bottom-up relabeling and
$A=(S,I,\Sigma',\Delta,a,R)$ is an $att$. 
The translation realized by $\hat{A}$ is given by
\[
\tau_{\hat{A}}=
\{ (s,t) \in T_\Sigma \times T_\Delta \mid \tau_B (s)= s' \text{ and } (s', t) \in \tau_{A}  \}.
\]
Functionality is defined for $atts^R$ in the obvious way. 
 We write
 $\tau_{\hat{A}} (s) =t$ as usual if $(s,t)\in \tau_{\hat{A}}$ if $\tau_{\hat{A}}$ is a function.
 An $att^R$ $\hat{A}=(B,A)$ is deterministic, i.e., a $datt^R$, if its underlying $att$ $A$ is.
If $A$ is a (deterministic) top-down transducer then  $\hat{A}$ is called a (deterministic) \emph{top-down transducer with look-ahead}
(or $(d)t^R$ for short).
We say that $\hat{A}=(B,A)$ is an $att^R$ with monadic output if $A$ is an $att$ with monadic output.\\

\textbf{Look-Around.}\quad
\emph{Look-Around} is similar to look-ahead; it is also a relabeling device that provides additional information to an $att$.
However, it is   more expressive than  look-ahead.
To define look-around, we first define  \emph{top-down relabelings}.
A top-down relabeling is a deterministic top-down tree transducer 
$T=(S,\emptyset,\Sigma ,\Sigma',a_0,R)$
where 
 all rules 
are of the form  $q (\sigma (x_1,\dots,x_k )) \rightarrow \sigma'(q_1( x_1),\dots,q_k (x_k))$
where $\sigma\in \Sigma_k$,  $\sigma'\in \Sigma'_k$ and $k\geq 0$.
Since  top-down relabelings are top-down transducers,  \emph{top-down relabelings with look-ahead} are defined in the obvious way.

An \emph{attributed tree transducer with look-around (or $att^U$)}
is a tuple $\breve{A}=(U,A)$
where $A$ is an $att$ and $U$ is a top-down relabeling with look-ahead.
The translation realized by $\breve{A}$ is defined analogously as for $att^R$.
This means that an $att^U$ relabels its input tree in two
phases: First the input tree is relabeled in a bottom-up phase. 
The resulting tree is relabeled again in a top-down phase before it is processed by~$A$.
Functionality and determinism  for $atts^U$ are defined analogously as for $atts^R$.
In particular, if $\breve{A}$ is deterministic then it is $datt^U$.
We say that $\breve{A}=(U,A)$ is an $att^U$ with monadic output if $A$ is an $att$ with monadic output.

The following results hold for $atts^U$.
First we show that any $datt^U$ $\breve{A}$ is of linear size increase, i.e.,
that a constant  $c\in \mathbb{N}$ exists such that
for all $(s,t)\in \tau_{\breve{A}}$, $|t| \leq c \cdot |s|$.

\begin{proposition}\label{datt lsi}
	Any $datt^U$ with monadic output is of linear size increase.
\end{proposition}

\begin{proof}
	Let $\breve{A}=(U,A)$ be a  $datt^U$ with monadic output.
	Obviously it is sufficient to show that the underlying $att$ $A$ of $\breve{A}$ is of linear size increase.
	Let $(s,t)\in \tau_A$. Then trees $t_1,\dots,t_n \in T_\Delta [SI(s^\#)]$ exist such that
	$
	a_0 (1) =  t_1\Rightarrow_{A, s^\#} \cdots \Rightarrow_{A, s^\#} t_n \Rightarrow_{A, s^\#} t.
	$
	Note that (even in the case that $A$ is circular) since $A$ is deterministic,
	for all $\alpha (\nu) \in \text{SI}(s^\#)$ at most one $j\in [n]$ exists
	such that $\alpha (\nu)$ occurs in $t_j$.  Clearly,
	if $\nu=\epsilon$ then  no $j\in [n]$ exists
	such that $\alpha (\nu)$ occurs in $t_j$.
	Thus, $n\leq  |S\cup I| \cdot |s|$. 
	Denote by $\text{maxsize}$ the maximal size of a right-hand side of a rule of $A$.
	Then clearly, $|t| \leq \text{maxsize}\cdot |S\cup I| \cdot |s| $.
\end{proof}

\noindent
With Proposition~\ref{datt lsi}, the following holds.

\begin{proposition}\label{equivalent proposition}
	For $datts^U$ with monadic output, equivalence is decidable.
\end{proposition}

\begin{proof}
	Let $\breve{A}$ be a $datt^U$ with monadic output.
	By Lemma~24
	\footnote{Note that (deterministic) $atts^U$
		are (deterministic)
		tree-walking tree transducers ($TTs$) given the definition of $TTs$ in~\cite{DBLP:journals/acta/EngelfrietIM21}.
		Specifically, the look-around in~\cite{DBLP:journals/acta/EngelfrietIM21}
		is equivalent to our look-around. More precisely, as stated in~\cite{DBLP:journals/acta/EngelfrietIM21} (before Lemma~10), the look-around of~\cite{DBLP:journals/acta/EngelfrietIM21} is  the same as a MSO-relabeling which is the same as a
		relabeling attribute grammar~\cite{DBLP:journals/jcss/BloemE00}~(Theorem~10), which in turn is equivalent to our look-around~\cite{DBLP:journals/iandc/EngelfrietM99}~(Theorem~4.4).}
	of~\cite{DBLP:journals/acta/EngelfrietIM21}
	a translation realized by  $\breve{A}$ can also be realized by deterministic macro tree transducers $M$.
	By Proposition~\ref{datt lsi}, $\breve{A}$ is of linear size increase.
	Since $\breve{A}$ and $M$ are equivalent, $M$ is obviously also of linear size increase. 
	By Corollary~13 of~\cite{DBLP:journals/ipl/EngelfrietM06}
	equivalence is decidable for deterministic
	macro tree transducers of linear size increase.
\end{proof}

\noindent
Finally, we prove the following result.

\begin{proposition}\label{empty test}
	Let $A$ be a  $att^U$. Let $L$ be a recognizable tree language.
	It is decidable whether or not $\text{range} (A) \cap L=\emptyset$.
\end{proposition}

\begin{proof}
	By Corollary~39 
	\footnote{Note that $atts$ are $0$-pebble tree transducers. Also note that by Theorem~13 of
		\cite{DBLP:journals/iandc/Baker79b},  any bottom-up tree transducer can be simulated by a composition of two top-down-tree transducers.}
	of~\cite{DBLP:journals/acta/EngelfrietM03}
	a composition of macro tree transducer $M$
	exists such that $\text{range} (A)=\text{range} (M)$. 
	Denote by $\tau_M$ the translation realized by $M$.
	Denote by $\tau_M^{-1} (L)$ the set $\{ s \mid \exists t \in L:\ (s,t)\in \tau_M\}$.
	By Theorem~7.4.1 of \cite{DBLP:journals/jcss/EngelfrietV85}, $\tau_M^{-1} (L)$  is recognizable.
	Thus emptiness is decidable for $\tau_M^{-1} (L)$. Obviously if $\tau_M^{-1} (L)$ is empty
	then $\text{range} (M)\cap L$ and hence  $\text{range} (A)\cap L$ are also empty.
\end{proof}

\section{From Attributed Tree Transducers with Monadic Output to Top-Down Tree Transducers}\label{monadic}
In the following section, 
we show that given a $datt$ with monadic output, i.e.,
a $datt$ where output symbols are  of rank at most~$1$, it is decidable whether or not an equivalent $dt^R$  exists.

\subsection{The Single Path Property}
Before we describe the decision procedure, consider the following definitions.

In the following, we
fix a $datt$ $A=(S,I,\Sigma,\Delta,a_0,R)$ with monadic output. 
For an input tree $s\in T_\Sigma$ and $v\in V(s)$, we say that on input $s$, an attribute $\alpha$ of $A$
\emph{processes} the node $v$ if a tree $t\in T_\Delta [\text{SI} (s^\#)]$ exists such that
$ a_0 (1)\Rightarrow^*_{A,s^\#} t$  and $\alpha (1.v)$ occurs in $t$.

Consider $s'\in T_\Sigma \cup T_{\Sigma^\#}$  and let $t, t' \in T_\Delta [\text{SI}(s')]$.
Then $t'$ is the \emph{normal form} of $t$ if $t \Rightarrow_{A,s'}^* t'$ and
no tree $t''$ exists such that $t' \Rightarrow_{A,s'} t''$.
We denote by $\text{nf} (\Rightarrow_{A,s'},t)$ the unique
 normal form of $t$ with
respect to $\Rightarrow_{A,s'}$ if it exists. Note that if $A$ is noncircular then
a unique normal form of $t$ always exists. However, if $A$ is circular then the existence of a
normal form is not guaranteed.

Consider an arbitrary $dt^R$ $\breve{T}=(B_T,T)$ with monadic output.
The behavior of~$T$ is limited in a particular way: Let $s$ be an input tree and let $B_T$ relabel $s$ into $s'$.
On input $s'$, the states of $T$ only process the nodes on a single
\emph{path} of $s'$. A path is a sequence of nodes $v_1,\dots,v_n$
such that $v_i$ is the parent node of $v_{i+1}$. 
This property follows as obviously at most one state occurs on the right-hand side of any rule of~$T$.
Using this property, we  show that
if a $dt^R$ $T=(B_T,T)$ equivalent to $A$ exists then 
a $datt^R$ $\hat{A}=(B,A')$ can be constructed from $A'$ such that
attributes of~$A'$ become limited in the same way as states of $T$: 
 Let $s$ be an input tree and let $B$ relabel~$s$ into~$\hat{s}$.
On input~$\hat{s}$, the states of $A'$ only process the nodes on a single
path of~$\hat{s}$.
 We call this property the \emph{string-like} property and call
$\hat{A}$ the $datt^R$ \emph{associated} with $A$.
Our proof uses of the 
result of~\cite{DBLP:conf/lics/FiliotGRS13}.
This result states that for a
\emph{two-way transducer} it is decidable whether or not an equivalent \emph{one-way transducer} exists.
Furthermore, in the affirmative case such an one-way transducer can be constructed.
Two-way transducers and one-way transducers
 are essentially attributed  transducers and 
top-down transducers with monadic input and monadic output, respectively.
We show that 
the $datt^R$ $\hat{A}$ associated with $A$ can be converted into a  two-way transducer $T_W$.
It can be shown that the  procedure of~\cite{DBLP:conf/lics/FiliotGRS13} yields
a one-way transducer $T_O$ equivalent to $T_W$ if and only if
a $dt^R$  equivalent to $A$ exists. 
Hence, it is decidable whether or not a $dt^R$ equivalent to $A$ exists.
We then show that in the affirmative case we can construct such a $dt^R$ from $T_O$.

Subsequently, we define the look-ahead with which we equip $A$.
Consider the rules of $A$.
Due to the following technical lemma, we assume  that  only right-hand sides of rules in $R_{\#}$ are ground (i.e., trees in $T_\Delta$).

\begin{lemma}\label{only}
	For any $att$ $A$ an equivalent $att$ $A'$ can be constructed such that
	 only its rules of $A'$ for the root marker have ground right-hand sides.
\end{lemma}

\begin{proof}
	Let  $A=(S,I,\Sigma,\Delta,a_0,R)$.
	We define $A'=(S,I',\Sigma,\Delta,a_0,R')$, where 
	\[
	I'= I\cup \{\langle \xi \rangle  \mid \xi \in T_\Delta \text{ is the right-hand side of some rule of } A\}.
	\]
	We define $R_\#' =R_\# \cup \{\langle \xi \rangle (\pi 1) \rightarrow \xi \mid \langle\xi \rangle \in I'\setminus I\}$.
	 Recall that as defined in Section~\ref{preliminaries},  $\pi 0 = \pi$. Let $k\geq 0$.
	 	For $\sigma\in \Sigma_k$, denote by $P(\sigma)$ the set of all rules $\rho\in R_\sigma$
	 which have ground right-hand sides.
	 Then we define
	\[
	\begin{array}{ll}
		R_\sigma'=&\{ \rho \mid \rho\in R_\sigma \setminus P (\sigma) \}\\
		&\cup \{\alpha (\pi i) \rightarrow \langle \xi \rangle (\pi) \mid 
		\alpha (\pi i) \rightarrow \xi \in P (\sigma), \text{ where } \alpha\in S\cup I\text{ and }i\geq 0\}\\
		&\cup \{\langle \xi \rangle (\pi j) \rightarrow \langle \xi \rangle (\pi)
		\mid \langle\xi \rangle \in I'\setminus I\ \text{and}\ 1\leq j\leq k  \}.
	\end{array}
\]
It should be clear that $A'$ and $A$ are equivalent. 
\end{proof}

Let $s\in \text{dom} (A)$ and let $v\in V(s)$.
We define the \emph{visiting pair set at $v$ on input $s$} as a subset of the set $\text{ISD}_A (s/v)$.
Informally, the visiting pair set at $v$ on input $s$
only contains those dependencies in $\text{ISD}_A (s/v)$
that actually occur at $v$ in the translation of $A$ on input $s$. 
More formally, 
we call the set $\psi\subseteq I\times S$ the 
\emph{visiting pair set at $v$ on input $s$}
 if	
	\[
	\psi =\{ (b,a)\in \text{ISD}_A (s/v) \mid \text{on input }s, \text{ the attribute }a \text{ of }A \text{ processes }v \}.
	\]
Let $\psi$ be the visiting pair set at $v$ on input $s$.
In the following, we denote by
$\Omega_\psi$ the set consisting of all trees $s'\in T_\Sigma$ such that $\text{ISD}_A (s') \supseteq \psi$.
This essentially means that the set $\Omega_\psi$ contains all trees $s'$ such that 
the visiting pair set at $v$ on input $s[v\leftarrow s']$ is also $\psi$.
If
$a\in S$ exists such that $(b,a)\in \psi$ for some $b\in I$ and 
the range of $a$ when translating trees in  $\Omega_\psi$
is unbounded, i.e., if
the cardinality of $\{\text{nf}(\Rightarrow_{A,s'}, a(\epsilon)) \mid s'\in \Omega_\psi \}$
is unbounded,
 then we say that the \emph{variation of $\Omega_\psi$} is unbounded.
 Note that for all $(b,a)\in \psi$ and all  $s'\in \Omega_\psi$, 
 $\text{nf}(\Rightarrow_{A,s'}, a(\epsilon))$
 is defined.
 Specifically, $(b,a)\in \psi$ and $s'\in \Omega_\psi$ implies $(b,a)\in \text{ISD}_A (s')$
 which by its definition and due to the fact that output symbols of  $A$ are of rank at most $1$
  implies that $\text{nf}(\Rightarrow_{A,s'}, a(\epsilon))$
 is defined.
  Note that obviously $b(\epsilon)$ occurs in $\text{nf}(\Rightarrow_{A,s'}, a(\epsilon))$.
  If $\psi$ is the visiting pair set at $v$ on input $s$ and the variation of $\Omega_\psi$
is unbounded then we also say that the \emph{variation at $v$ on input $s$} is unbounded.

The variation plays a key role for proving our claim. 
In particular, the following property is derived from it:
We say that $A$ has the \emph{single path property} if for all trees $s \in dom(A)$ 
a path $\rho$ exists such that the variation at $v \in V(s)$ is bounded whenever $v$ does
not occur in $\rho$.  The following lemma states that the single path property
is a necessary condition for the $att$ $A$ to have an equivalent $dt^R$.

\begin{lemma}\label{necessary condition}
	If a $dt^R$ $T$ equivalent to $A$ exists then $A$ has the single path property. 
\end{lemma}
\begin{proof}
	Denote by $B_T$ the bottom-up relabeling of $T$.
	Let $l_1,\dots,l_n$ be the states of $B_T$.
	W.l.o.g. assume that $B_T$ operates as follows:
	Consider the tree $\sigma (s_1,\dots,s_k)$, where $\sigma\in \Sigma_k$
	and $s_1,\dots,s_k \in T_\Sigma$.
	Then the bottom-up relabeling $B_T$ of $T$ relabels
	$\sigma$  by the symbol $\sigma_{l_1',\dots,l_k'}$ by  if
	$s_i \in\text{dom}_{B_T} (l_i')$ for $i\in [k]$.

	Let $s\in\text{dom} (A)$ and  $v_1, v_2 \in V(s)$ such that 
	$v_1$ and $v_2$ have the same parent node and $v_1\neq v_2$.
	Assume that  the variation at both
$v_1$ and $v_2$ is unbounded. Let $\psi_1$ and $\psi_2$ be the visiting pair sets
at $v_1$ and $v_2$ on input $s$, respectively.
	Since $T$ and $A$ are equivalent, we can assume that $\text{dom} (B_T)=\text{dom} (A)$. 
	Note that by~\cite{DBLP:journals/ipl/FulopM00}, the domain of $A$ is  effectively recognizable.
	Since  $s\in\text{dom} (A)$,
	it holds that
	 $s[v_1\leftarrow s_1]\in \text{dom} (A)$
	for all $s_1 \in \Omega_{\psi_1}$. 
	This in turn means that for all $s_1 \in \Omega_{\psi_1}$ a state $l$ of $B_T$
	exist such that  $s_1\in \text{dom}_{B_T} (l)$.
	Therefore, all non-empty sets of the form
	$
	\Omega_{\psi_1} \cap \text{dom}_{B_T} (l_i),
	$
	where $i\in [n]$,
	form a partition of $\Omega_{\psi_1}$.
	Hence, and as the variation of
	$\Omega_{\psi_1}$ is unbounded, a state $l_{i_1}$ of $B_T$
	and $(b,a)\in \psi_1$ exist
	such that the cardinality of
	\[
	\mathcal{U}_1=\{\text{nf}(\Rightarrow_{A,s}, a(\epsilon))\mid  s\in \Omega_{\psi_1} \cap \text{dom}_{B_T}
	(l_{i_1}) \}
	\]
	is unbounded.
	Analogously, it follows that a state $l_{i_2}$ of $B_T$ 
	with the same property
	exists for  
	$\Omega_{\psi_2}$. 
	
	In the following let $s_j\in \Omega_{\psi_j} \cap \text{dom}(l_{i_j})$ for $j=1,2$.
	Consider the tree 
	$
	\hat{s}=s[v_j \leftarrow s_j\mid j=1,2].
	$
	Note that the visiting pair sets at $v_1$ and $v_2$ on input $\hat{s}$ are also
	$\psi_1$ and $\psi_2$, respectively. 
	Thus,  $\hat{s}\in \text{dom} (A)$ and consequently $\hat{s}\in \text{dom} (T)$.
	Let $T$ produce $t\in T_\Delta$ on input $\hat{s}$.
	As $T$ produces monadic output trees, it is obvious that on input $\hat{s}$ either $v_1$ or $v_2$
	is not processed by a state of $T$. W.l.o.g. assume that $v_1$ is not processed by
	$T$.
	Now consider an  arbitrary  $s_1'\in \Omega_{\psi_1} \cap \text{dom}_{B_T}(l_{i_1})$.
	Since $s_1'\in  \text{dom}_{B_T}(l_{i_1})$ and as $T$ is deterministic
	it follows 
	that on inputs $\hat{s}$ and $\hat{s}[v_1\leftarrow s_1']$ the same output
	tree is produced by $T$, i.e., for both input trees the output tree $t$ is produced.
	However, depending on the choice of $s_1'$,
	$A$ does \emph{not} produce the same output tree
	on inputs $\hat{s}$ and $\hat{s}[v_1\leftarrow s_1']$.
	Because the unboundedness of the set
	$\mathcal{U}_1$,
	a tree $\tilde{s}_1\in  \Omega_{\psi_1} \cap \text{dom}(l_{i_1})$
	and a pair $(b,a)\in \psi_1$
	exist such that
	\[
	\text{height} (t) < \text{height} (\text{nf}(\Rightarrow_{A, \tilde{s}_1},a(\lambda))).
	\]
	Consider the tree $\hat{s}[v_1\leftarrow \tilde{s}_1]$.
	Since the visiting pair set at $v_1$ on input $\hat{s}[v_1\leftarrow \tilde{s}_1]$ is also $\psi_1$,
	it follows easily that on input  $\hat{s}[v_1\leftarrow \tilde{s}_1]$,
	$A$ yields a tree of height greater than $\text{height} (t)$.
	Therefore, $A$ and $T$ are not equivalent. 
\end{proof}

\begin{example}\label{att example 2}
	Consider the $att$ $A_2=(S,I,\Sigma,\Delta,a,R)$ where
	$\Sigma= \{f^2, e^0, d^0\}$ and $\Delta=\{f^1, g^1, e^0, d^0\}$.
	The set of attributes are given by $S=\{a, a_e,a_d\}$ and $I=\{b_e, b_d, \langle e\rangle, \langle d \rangle\}$.
	We define
	\[
	\begin{array}{cc ccc c ccc c ccc}
	R_f= \{\ & a_d (\pi) & \rightarrow & f (a (\pi 1)  ), & \quad &   	b_d (\pi 1) & \rightarrow & a_d (\pi 1), &\quad & 	b_d(\pi 2) &  \rightarrow & b_d(\pi), \\ 
	& a_e (\pi) & \rightarrow & g (a (\pi 1)  ),  & \quad & b_e (\pi 1) & \rightarrow & a_e (\pi 1),  &\quad & 	b_e(\pi 2) & \rightarrow & b_e(\pi), \\
	& a(\pi) & \rightarrow & a(\pi 2),  & \quad &  \langle e \rangle (\pi 1) & \rightarrow & \langle e \rangle (\pi ),  &\quad & \langle d \rangle (\pi 1) & \rightarrow & \langle d \rangle (\pi )\ \} \\
	\end{array}
\]
\[
\begin{array}{cc ccc c ccc c ccc ccc}
\text{and }R_\#=\{ &
b_e(\pi 1) &  \rightarrow & a_e (\pi 1),  &\ & b_d(\pi 1) & \rightarrow  & a_d (\pi 1), &\
\langle e \rangle (\pi 1) & \rightarrow &e, &\ & \langle d \rangle (\pi 1) & \rightarrow & d\ \}.\\
\end{array}
	\]
	Furthermore, we define  
	\[
	R_e=\{a(\pi)\rightarrow b_e (\pi),\ a_e (\pi)\rightarrow \langle e \rangle (\pi)\}
	\text{ and }
	R_d=\{a(\pi)\rightarrow b_d (\pi),\ a_d (\pi)\rightarrow \langle d \rangle (\pi)\}.
	\]
	Let $s\in T_\Sigma$ and denote by $n$ the length of the leftmost path of $s$.
	On input~$s$, $A_2$ produces the tree $t$  of height $n$ whose nodes are labeled as follows:
	if $v\in V(t)$ is not a leaf and the rightmost leaf of the subtree $s/v$ is labeled by $e$ then 
	$t[v]=g$, otherwise $t[v]=f$.
	If $v$ is a leaf then $t[v]=s[v]$. 
		For instance, the input tree $s=f ( f( f(d,d),d), f(d,e) )$   
		is translated to the output tree
		$g(f(f(d)))$. 
\end{example}

Clearly, a  $dt^R$ that is equivalent to the $att$ $A_2$ in Example~\ref{att example 2} exists. 
Furthermore $A_2$ has the single path property. In particular, it can be verified that the
variations of all nodes that do not occur on the left-most path of the input tree
are bounded. More precisely, if the node $v$ does not occur on the leftmost path of the input tree
then its visiting pair set is either $\psi_e=\{ (b_e,a)\}$ or $\psi_d=\{ (b_d,a)\}$.
Consider the set $\Omega_{\psi_e}$.
Clearly, $\Omega_{\psi_e}$ consists of all trees in $T_\Sigma$ whose rightmost leaf is labeled by $e$.
For all such trees the attribute $a$ yields the output $b_e (\epsilon)$. This in turn means
that the variation of $\Omega_{\psi_e}$ is bounded.
The case for $\psi_d$ is analogous.

In contrast, consider the $att$ $A_1$ in Example~\ref{att example}. 
Recall that it
translates an input tree $s$ into a monadic tree $t$ of height  $n+1$ if $s$ has $n$ leaves.
This translation is not realizable by any $dt^R$. This is reflected in the fact that
the $att$ of Example~\ref{att example} does \emph{not} have
the single path property.
In particular, consider $s=f(f(e,e), f(e,e))$.
The visiting pair set at all nodes of $s$  is $\psi=\{(a,b)\}$.
Furthermore, $\Omega_\psi=T_\Sigma$.
It can be verified that the variation of $\Omega_\psi$  is unbounded.

Recall that we aim to construct the 
$att^R$ $\hat{A}=(B,A')$ associated with $A$ and that we require $\hat{A}$ to have the string-like-property.
This property is closely related to the
single path property.
In particular, the basic idea behind $\hat{A}$ is as follows. 
Let $s\in \text{dom} (A)$ and let $B$ relabel $s$ into $s'$.
The idea is that on input $s'$, if attributes of $A'$  process $v\in V(s')$ then
the variation at $v$ on input $s$ with respect to $A$ is unbounded. 
Note that obviously $V(s')=V(s)$.
Clearly, if  $A$ has the single path property then attributes of $A'$
only process nodes of a single path of $s'$.

Now the question is how precisely do we construct $\hat{A}$? 
To construct $\hat{A}$, the basic idea is to precompute all
parts of the output tree
that would be otherwise produced at nodes with bounded variation
using the look-ahead of $\hat{A}$.
To do so, we first require the following lemma.

\begin{lemma}\label{visitin pair sets computable}
The set of all visiting pair sets of
$A$  can be computed.
\end{lemma}

\begin{proof}
To compute all subsets of $I\times S$ that are visiting pair sets of
$A$,
we record which attribute of $A$ process which
node of the input tree. 
To do so, we construct the $att$ $A'$ from $A$.
The general idea is as follows: Let $s\in T_\Sigma$.
We mark a single node $v$ of $s$ by annotating its label by $\pm$
to distinguish it.
Whenever that node $v$ is processed by an attribute~$\alpha$, we record $\alpha$ by making $\alpha$ part of the output.

Formally, we define $A'=(S,I,\Sigma',\Delta',R',a_0)$
with $\Sigma'=\Sigma \cup \{\sigma_\pm \mid \sigma \in \Sigma\}$ where $\sigma_\pm \in \Sigma'_k$ if $\sigma\in \Sigma_k$
and $\Delta'= \{ \breve{\alpha} \mid \alpha\in S\cup I\}\cup \{e\}$
where  $e$ is of rank $0$ and all remaining symbols in $\Delta'$
are of rank $1$.
We demand that input trees of $A'$ contain exactly one
node labeled by a symbol of the form $\sigma_\pm$.

We now define the rule set $R'$.
Recall that due to Lemma~\ref{only}, we assume that only rules of $A$ for the root marker have ground right-hand sides.
We define that if
$b(\pi 1)\rightarrow t\in R_{\#}$, where $t$ is ground,
then $b(\pi 1)\rightarrow e \in R'_{\#}$.
The remaining rules of $A'$ are obtained as follows:
Let  $\sigma \in \Sigma\cup \{\#\}$
and let  $\rho\in R_\sigma$.
In case that $\sigma=\#$ let the right-hand side of $\rho$ be non-ground.
Then, we define that  $\rho' \in R_\sigma'$, where $\rho'$ is obtained from $\rho$
by removing all
$\Delta$-symbols occurring in $\rho$, e.g.,
if $R_\sigma$ contains $a_1(\pi)\rightarrow f(g(a_2(\pi 1)))$ and $b (\pi 1) \rightarrow g(a_2(\pi 1))$ 
where $a_1,a_2\in S$, $b\in I$ and $f,g\in \Delta$ then
$a_1(\pi)\rightarrow a_2(\pi 1), b (\pi 1) \rightarrow a_2(\pi 1)\in R_\sigma'$.
For a symbol of the form $\sigma_\pm$, where $\sigma\in \Sigma_{k}$, we proceed as follows:
Let  $a(\pi) \rightarrow t \in R_{\sigma}$,
where $a\in S$. 
\begin{enumerate}
	\item If $a' (\pi j)$ with  $a' \in S$ and $j\in [k]$ occurs in $t$
	then
	$
	a(\pi) \rightarrow \breve{a} (a' (\pi j))\in R'_{\sigma_\pm}.
	$
	\item If $b(\pi)$ occurs in $t$ 
	where  $b\in I$ then
	$
	a(\pi) \rightarrow \breve{a} (\breve{b}  (b (\pi)))\in R'_{\sigma_\pm}.
	$
\end{enumerate}
Let $b(\pi j) \rightarrow t'\in R_{\sigma}$ where $b\in I$ and $j\in [k]$.
\begin{enumerate}
\item If $b' (\pi)$ occurs in $t'$
where $b'\in I$   then
$
b(\pi j) \rightarrow \breve{b}' (b' (\pi))\in R'_{\sigma_\pm}.
$
\item If $a (\pi i)$ occurs in $t'$
where $a\in S$ and $i\in [k]$ then
$
b(\pi j) \rightarrow  a (\pi i)\in R'_{\sigma_\pm}.
$
\end{enumerate}
Note that so far $A'$ does not test whether or not  exactly one symbol of the input tree is marked.
To address this issue,
we equip $A'$ with a deterministic bottom-up relabeling $B'$ thus obtaining
the $att^R$ $(B',A')$.
On input $s$, $B'$ tests whether or not the tree $s$ is of the from
we demanded. If so then $B'$ outputs $s$; otherwise
no output is produced. Clearly such a bottom-up relabeling $B'$ can be constructed.

Let $s\in T_\Sigma$ and $s'$ be obtained from $s$ by `marking' a single node $v$
as specified earlier.
By construction of $A'$, it is clear that
in a translation of $A$ on input $s$, the attribute $\alpha$ processes the node $v$ of $s$
if and only if $\alpha$ also processes  $v$
in a translation of $A'$ on input $s'$.
Furthermore, clearly $s\in \text{dom} (A)$ if and only if $s'\in \text{dom} (A')$.
If $s'\in \text{dom} (A')$ then $A'$ outputs a
tree  of the form
$\breve{a}_1 (\breve{b}_1( \cdots (\breve{a}_n( \breve{b}_n ( e))) \cdots ))$,
where $a_1,\dots, a_n \in S$ and $b_1,\dots,b_n \in I$, on input $s'$. Clearly, this means that for a translation of $A$, 
the visiting pair set at $v$ on input $s$ is
$\{ (a_1,b_1),\dots, (a_n,b_n)\}$.

Thus, we obtain the set of all visiting pair sets of $A$
by computing the range of the $att^R$ $(B',A')$.
Note that since  $(B',A')$ is deterministic, it cannot produce output trees 
$\breve{a}_1 (\breve{b}_1( \cdots (\breve{a}_n( \breve{b}_n ( e))) \cdots ))$
such that for some $i,j\in [n]$, where $i\neq j$, 
either $\breve{a}_i = \breve{a}_j$ or $\breve{b}_i = \breve{b}_j$.
Therefore, output trees of $(B',A')$ have at most $|S|+|I|+1$ nodes and hence
the range of $(B',A')$ is bounded.
Thus it can be computed
due to Theorem~4.5 of~\cite{DBLP:journals/iandc/DrewesE98}, see also
Lemma~3.8 of~\cite{DBLP:journals/siamcomp/EngelfrietM03}.
\end{proof}

\noindent
Additionally, we require the following lemma as well. 

\begin{lemma}\label{bounded computable}
	Let $\psi$ be  the visiting pair set for some tree $s$ and some node $v\in V(s)$. It is decidable
	whether or not the variation of $\Omega_\psi$ is bounded. 
\end{lemma}

\begin{proof}
In the following, we construct for all $(b,a) \in \psi$
an $att^R$ $\breve{A}_{b,a}$=$(B,A(a))$.
The $att$ $A(a)$ is obtained from $A$ by replacing
the initial attribute of $A$ by $a$
and replacing
 the rule set $R_\#$ of $A$ by
the set $\{b(\pi 1)\rightarrow e\}$
where $e$ is some symbol of $\Delta$ of rank $0$.
Recall that due to Lemma~\ref{only}, we assume that only rules of $A$ for the root marker have ground right-hand sides.

The relabeling $B$ basically tests whether or not
an input tree $s$ is an element of
$\Omega_{\psi}$. If so, then $B$ outputs $s$ on input $s$;
otherwise no output is produced.
To construct $B$, 
recall that the is-dependencies of~$A$
can be computed in a bottom-up fashion, i.e.,
a
bottom-up-relabeling 
$\breve{B}$ whose states are subsets of $I\times S$ can be constructed such that 
for all $s\in T_\Sigma$, $s\Rightarrow_{\breve{B}}^* (ISD_A(s)) (s)$.
Note that the input and output alphabet of $\breve{B}$ are identical.
Furthermore all states of $\breve{B}$ are final states.
Formally, $B$ is obtained from $\breve{B}$
by defining that only states $p$ of $\breve{B}$ such that $\psi \subseteq p$
are final states.

Obviously, the variation of $\Omega_{\psi}$ is bounded if and only if 
for all $(b,a)\in \psi$,
the range of the corresponding
$att^R$ $\breve{A}_{b,a}$ is finite.
By Theorem~4.5 of~\cite{DBLP:journals/iandc/DrewesE98}, see also
Lemma~3.8 of~\cite{DBLP:journals/siamcomp/EngelfrietM03}, finiteness of ranges is decidable.
\end{proof}

Let $\psi$ be  a visiting pair set and
let the variation of $\Omega_\psi$ be bounded. 
Then a minimal integer $\kappa_{\psi}$ can be computed such that
 for all $a\in S$ such that $(b,a)\in \psi$ for some $b\in I$ and for all $s'\in \Omega_\psi$,
 $\text{height} ( \text{nf} (\Rightarrow_{A,s'}, a(\epsilon)) )\leq \kappa_{\psi}$.
More precisely, to compute $\kappa_{\psi}$, we simply need to compute the ranges of all the $atts^R$
constructed in the proof of Lemma~\ref{bounded computable}.
The ranges of these $atts^R$ are computable according to Theorem~4.5 of~\cite{DBLP:journals/iandc/DrewesE98}. 
Therefore it follows  along with Lemmas~\ref{visitin pair sets computable} and~\ref{bounded computable} that 
\[
\kappa=\text{max} \{\kappa_\psi \mid\psi \text{ is a visiting
pair set of } A  \text{ and the variation of }\Omega_\psi\text{ is bounded}\}
\]
is computable. 
Denote by 
$T_\Delta^\kappa [I(\{\epsilon\})]$ the set of all trees in $T_\Delta [I(\{\epsilon\})]$  
that are of height at most $\kappa$. Informally, the bottom-up relabeling $B$ constructed subsequently,
precomputes output subtrees of height at most $\kappa$ that contain the inherited attributes of the root of the current input subtree. Hence, the $att$ $A'$ of $\hat{A}$ does not need to compute those 
output subtrees itself; the translation  is continued immediately with  those 
output subtrees.

Formally,  $B=(\mathcal{P},\Sigma,\Sigma^B,P,R_B)$ where states in $\mathcal{P}$ are sets 
\[
\varrho\subseteq\{ (a,\xi) \mid a\in S\text{ and } \xi \in T_\Delta^\kappa [I(\{\epsilon\})]\}.
\]
The idea is that if
$s\in \text{dom}_{B} (\varrho)$ and
$(a,\xi) \in \varrho$ then $\xi=\text{nf} (\Rightarrow_{A,s}, a(\epsilon))$.
Conversely, let $a'$ be a synthesized attribute for which no tree $\xi'\in T_\Delta [I(\{\epsilon\})]$
exists such that $(a',\xi') \in \varrho$.
For such an attributes $a'$ and all $s\in \text{dom}_{B} (\varrho)$,
 the height of
$\text{nf} (\Rightarrow_{A,s}, a'(\epsilon))$  either exceeds $\kappa$ or
$\text{nf} (\Rightarrow_{A,s}, a'(\epsilon))$ is not in  $T_\Delta [I(\{\epsilon\})$.

The symbols in $\Sigma^B$ are of the form $\sigma_{\varrho_1,\dots,\varrho_k}$
where $\sigma \in \Sigma_k$ and $\varrho_1,\dots,\varrho_k \in \mathcal{P}$.
Let $\sigma (s_1,\dots,s_k)\in T_\Sigma$. Then $B$ relabels $\sigma$ by $\sigma_{\varrho_1,\dots,\varrho_k}$
if for $i\in [k]$, $s_i \in \text{dom}_{B} (\varrho_i)$.
To do so the rules of $B$ are defined as follows:
Let $\sigma \in \Sigma_k$, states $\varrho_1,\dots,\varrho_k\in\mathcal{P}$ 
and the rules in the set $R_\sigma$ of $A$ be given.
For all $i\in [k]$, we fix a tree $s_i'$ such that $s_i' \in \text{dom}_{B} (\varrho_i)$.
We remark that the following definition does not depend on the choice of $s_i'$. 
Given the trees $s_1',\dots, s_k'$  and  $R_\sigma$ we compute
the set $\varrho$ containing all pairs $(a,\xi)$
such that $\xi=\text{nf} (\Rightarrow_{A,\sigma(s_1',\dots,s_k')} ,a(\epsilon) )$ is a tree in $T_\Delta^\kappa [I(\{\epsilon\})]$.
With $\varrho$,  we define that 
$
\sigma (\varrho_1 (x_1),\dots, \varrho_k (x_k))\rightarrow \varrho (\sigma_{\varrho_1,\dots,\varrho_k} (x_1,\dots,x_k)) \in R_B.
$

\begin{example}\label{att example 3}
	Consider the $att$ $A_2$ in Example~\ref{att example 2}.
	Recall that all nodes that do not occur on the leftmost path of the input tree $s$ of $A_2$
	have bounded variation. Let $v$ be such a node. Then the visiting pair set
	at $v$ is either $\psi_e=\{ (a,b_e)\}$ or $\psi_d=\{ (a,b_d)\}$. 
	Assume the former. Then $\text{nf} (\Rightarrow_{A_2,s/v}, a(\epsilon))=b_e(\epsilon)$.
	If we know beforehand that $a$ produces $b_e(\epsilon)$ when translating $s/v$, then there is no need to process $s/v$ with $a$ anymore.
	This can be achieved via a bottom-up relabeling $B_2$ that precomputes all output trees of
	height at most $\kappa=\kappa_{\psi_e}= \kappa_{\psi_d}=1$.
	In particular the idea is that
	if for instance $v\in V(s)$ is  relabeled by $f_{\{(a,b_d(\epsilon))\}, \{a, b_e(\epsilon)\}}$
	then this means when translating
	$s/v.1$ and $s/v.2$, $a$ produces $b_d(\epsilon)$ and $b_e (\epsilon)$, respectively.
	For completeness, the full definition of $B_2$ is as follows:
	The states of $B_2$ (which are all also final states) are 
	\[
	\begin{array}{ccc}
		\varrho_1= \{ (a_e, \langle  e\rangle (\epsilon)),\ (a, b_e (\epsilon)) \} & \quad & \varrho_3= \{(a,b_d(\epsilon))\} \\
		\varrho_2= \{ (a_d, \langle  d \rangle (\epsilon)),\ (a, b_d (\epsilon))  \} & \quad & \varrho_4= \{ (a, b_e(\epsilon)) \}.
	\end{array}
	\]
	In addition to $e \rightarrow  \varrho_1 (e)$ and $d \rightarrow  \varrho_2 (d)$, $B_2$ also
	contains the rules
	\[
	\begin{array}{cccc ccc }
		f(\varrho (x_1), \varrho_1 (x_2)) & \rightarrow & \varrho_4 (f_{\varrho, \varrho_1} (x_1,x_2)) &\quad &
		 f(\varrho (x_1), \varrho_2 (x_2)) & \rightarrow & \varrho_3 (f_{\varrho, \varrho_2} (x_1,x_2))  \\
		f(\varrho (x_1), \varrho_3 (x_2)) & \rightarrow & \varrho_3 (f_{\varrho, \varrho_3} (x_1,x_2)) &\quad&
		 f(\varrho (x_1), \varrho_4 (x_2)) & \rightarrow & \varrho_4 (f_{\varrho, \varrho_4} (x_1,x_2)),
	\end{array}
	\]
	where $\varrho\in \{ \varrho_1,\dots,\varrho_4\}$.
	Using $B_2$ we will later construct the $att$ $A_2'$, that is, $A_2$ modified to make
	use of $B_2$ so that only nodes of the leftmost path
	of~$s_2$ are processed.
\end{example}

To construct  the $att^R$ $\hat{A}=(B,A')$, all that is left is to define $A'$.
We define  $A'=(S,I,\Sigma^B ,\Delta, a_0, R')$.
The rules of $A'$ for a symbol $\sigma_{\varrho_1,\dots,\varrho_k}\in \Sigma^B$
 are defined as follows.
First, we define for each 
state $\varrho$ of $B$ an auxiliary symbol $\langle \varrho \rangle$ of rank~$0$ with which we expand $A$.
For such a symbol, we define the rule
 $a(\pi)\rightarrow t \in R_{\langle \varrho \rangle}$ if the pair $(a,t[\pi\leftarrow \epsilon])$ occurs $\varrho$. Recall that $t[\pi\leftarrow \epsilon]$  denotes the substitution that replaces all occurrences
of $\pi$ in $t$ by  $\epsilon$.
Now consider the tree $\sigma(\langle\varrho_1 \rangle,\dots, \langle\varrho_k \rangle)$ and
let 
$
\eta=\text{nf} (\Rightarrow_{A, \sigma(\langle\varrho_1 \rangle,\dots, \langle\varrho_k \rangle)},a(\epsilon))$
be a tree in $T_\Delta [I(\{\epsilon\}) \cup S([k])]$.
Then we define the rule
$a(\pi)\rightarrow \eta'\in R'_{\sigma_{\varrho_1,\dots,\varrho_k}}$ for $A'$ where $\eta'$
is the tree such that $\eta'[\pi\leftarrow \epsilon] =\eta$.
Finally, we define $R'_\# =R_\#$.
It should be clear that the following holds.

\begin{lemma}\label{associate equal}
	The $att$ $A$ and its associated $att^R$ $\hat{A}=(B,A')$ are equivalent.
\end{lemma} 

Furthermore, since all output subtrees of height at most $\kappa$
are precomputed, attributes of $A'$  only process nodes whose
variation  with respect to $A$ are unbounded. To illustrate this point, consider the following example.

\begin{example}\label{att example 4}
Given the relabeling $B_2$ constructed in Example~\ref{att example 3}, we first construct the $att$ 
$A_2'$ from the $att$ $A_2$ of Example~\ref{att example 2} so that we can make use of $B_2$.
To begin with, the input alphabet of $A_2'$ consists of all output symbols of $B_2$.
In particular, note that the output alphabet of $B_2$ contains all input symbols of $A_2$ that are of rank~$0$.
Thus, the rules of $A_2'$ are defined as follows: First of all, every rule of $A_2$ for a symbol of rank $0$
or the root marker is carried over to $A_2'$.
The remaining rules of $A_2'$ are defined as follows:
The set $R_{f_{\varrho_1,\varrho_2}}$ contains the rules
\[
\begin{array}{ccc c ccc c ccc}
	 a_d (\pi) & \rightarrow & f (\langle e \rangle (\pi)) & \quad &   	b_d (\pi 1) & \rightarrow & a_d (\pi 1) &\quad & 	b_d(\pi 2) &  \rightarrow & b_d(\pi) \\
	                            a_e (\pi) & \rightarrow & g (\langle e \rangle (\pi))  & \quad & b_e (\pi 1) & \rightarrow & \langle e \rangle (\pi)  &\quad & 	b_e(\pi 2) & \rightarrow & b_e(\pi) \\
		                        
		                      a(\pi) & \rightarrow & b_d(\pi)\ & \quad &  \langle e \rangle (\pi 1) &\rightarrow &   \langle e \rangle (\pi) & \quad & \langle d \rangle (\pi 1) &\rightarrow & \langle d \rangle (\pi), \\
\end{array}
\]
 the set $R_{f_{\varrho_2,\varrho_2}}$ contains the rules
 \[
 \begin{array}{ccc c ccc c ccc }
 a_d (\pi) & \rightarrow & f (\langle d \rangle (\pi)) & \quad &   	b_d (\pi 1) & \rightarrow & \langle d \rangle (\pi) &\quad & 	b_d(\pi 2) &  \rightarrow & b_d(\pi)\\ 
 	                            a_e (\pi) & \rightarrow & g (\langle d \rangle (\pi))  & \quad & b_e (\pi 1) & \rightarrow &  a_e (\pi 1) &\quad & 	b_e(\pi 2) & \rightarrow & b_e(\pi) \\
 		                        a(\pi) & \rightarrow & b_d(\pi) & \quad &  \langle e \rangle (\pi 1) &\rightarrow &   \langle e \rangle (\pi) & \quad & \langle d \rangle (\pi 1) &\rightarrow & \langle d \rangle (\pi),\\
 \end{array}
 \]
 the set $R_{f_{\varrho_3,\varrho_2}}$ contains the rules
 \[
 \begin{array}{l ccc c ccc c ccc }
  	a_d (\pi) & \rightarrow & f (a_d (\pi1)) & \quad &   	b_d (\pi 1) & \rightarrow & a_d(\pi 1) &\quad & 	b_d(\pi 2) &  \rightarrow & b_d(\pi)\\
 	                             a_e (\pi) & \rightarrow & g (a_d (\pi1))  & \quad & b_e (\pi 1) & \rightarrow &  a_e (\pi 1) &\quad & 	b_e(\pi 2) & \rightarrow & b_e(\pi) \\
 		                          a(\pi) & \rightarrow & b_d(\pi) & \quad &  \langle e \rangle (\pi 1) &\rightarrow &   \langle e \rangle (\pi) & \quad & \langle d \rangle (\pi 1) &\rightarrow & \langle d \rangle (\pi) \\
 \end{array}
 \]
 and the set $R_{f_{\varrho_4,\varrho_2}}$ contains the rules
 \[
 \begin{array}{ccc c ccc c ccc}
 	 a_d (\pi) & \rightarrow & f (a_e (\pi1)) & \quad &   	b_d (\pi 1) & \rightarrow & a_d(\pi 1) &\quad & 	b_d(\pi 2) &  \rightarrow & b_d(\pi)\\  
 	                               a_e (\pi) & \rightarrow & g (a_e (\pi1))  & \quad & b_e (\pi 1) & \rightarrow &  a_e (\pi 1) &\quad & 	b_e(\pi 2) & \rightarrow & b_e(\pi) \\
 	                               a(\pi) & \rightarrow & b_d(\pi) & \quad &  \langle e \rangle (\pi 1) &\rightarrow &   \langle e \rangle (\pi) & \quad & \langle d \rangle (\pi 1) &\rightarrow & \langle d \rangle (\pi).\\
 \end{array}
 \]
 We remark that for $\varrho_i$ with $1\leq i \leq 4$, the rule sets $R_{f_{\varrho_i,\varrho_2}}$
 and $R_{f_{\varrho_i,\varrho_3}}$ are identical.
 As for the remaining rules, the rule sets $R_{f_{\varrho_i,\varrho_1}}$ and $R_{f_{\varrho_i,\varrho_4}}$ are obtained from $R_{f_{\varrho_i,\varrho_2}}$
 by replacing the rule $a(\pi) \rightarrow b_d(\pi)$ by  $a(\pi) \rightarrow b_e(\pi)$.

 This concludes the construction of the $att^R$ $\hat{A}_2=(B_2, A_2')$.
It is easy to see that on input $s'$, i.e., the tree obtained from $s\in T_\Sigma$ via the relabeling $B_2$, attributes of $A_2'$ only process nodes occurring on the left-most
path of $s'$.
\end{example}

Recall that by Lemma~\ref{necessary condition}, the single path property is a necessary condition for
the existence of a $dt^R$ equivalent to $A$. We will now show how to test whether $A$ has the single path property using its associated
$att^R$ $\hat{A}=(B,A')$.

\begin{lemma}\label{necessary condition decidable}
	It is decidable whether or not $A$ has the single path property.
	In the affirmative case, its associated
	$att^R$ $\hat{A}=(B,A')$ has the string-like property.
\end{lemma}

\begin{proof}
Consider the $att^R$ $\hat{A}=(B,A')$ associated with $A$.
Recall that by Lemma~\ref{associate equal}, $\hat{A}$ and $A$ are equivalent.
Let  $s\in \text{dom} (A)=\text{dom} (\hat{A})$ and let $B$ relabel $s$ into $s'$.
By construction of $\hat{A}$, if nodes $v_1, v_2 \in V(s')$ with the same parent node
exist such that on input $s'$, attributes of $A'$ process both $v_1$ and $v_2$
then  $A$ does not have the single path property.
Thus, to test whether $A$ has the single path property,
we construct the following $att^R$ $\breve{A} = (\breve{B},\breve{A}')$ from $\hat{A}=(B,A')$.
The idea is similar to the idea in the proof in Lemma~\ref{visitin pair sets computable}.
Input trees of $\breve{A}$ are trees
$s\in \text{dom} (\hat{A})$ where two nodes $v_1, v_2$ with the same parent node
are annotated by flags $f_1$ and $f_2$ respectively.
The relabeling $\breve{B}$  checks whether or not
the flags $f_1$ and $f_2$ occur both exactly once in the input tree $s$
and whether or not the nodes $v_1$ and $v_2$ at which these flags appear have the same  parent node.
If not then $\breve{B}$ produces no output.
Additionally, $\breve{B}$ relabels input nodes as $B$ would, where nodes annotated with flags are relabeled in
the obvious way.
The $att$ $\breve{A}'$ simulates $A'$ such that
output symbols are only produced if an annotated symbol is processed by a synthesized attribute
or if a rule, where the right-hand side is ground, is applied.
In particular, for $i=1,2$ we introduce a special symbol $g_i$ which is only outputted
if the node with the flag $f_i$ is processed.
Hence, we simply need to check whether there is a tree with occurrences of both  
$g_1$ and $g_2$  in the range of $\breve{A}$.
By construction of the rules of $\breve{A}$, the range of $\breve{A}$ is finite. Thus it can be computed~\cite{DBLP:journals/iandc/DrewesE98,DBLP:journals/siamcomp/EngelfrietM03}.
\end{proof}

\subsection{From Tree to String Transducers and Back}\label{tree to string}
In the following let $\hat{A}=(B,A')$ be a fixed $att^R$ with the string-like property.
In this section we show how to construct an equivalent $dt^R$ if it exists. 
To do so, we construct a two-way transducer $T_W$ from $\hat{A}$ such that
a one-way transducer $T_O$ equivalent to $T_W$ exists if and only if 
a $dt^R$ equivalent to $\hat{A}$ exists.
Thus, due to the procedure in~\cite{DBLP:conf/lics/FiliotGRS13}, it is decidable 
whether or not a $dt^R$ equivalent to $\hat{A}$. Finally we show how to construct such a $dt^R$ from $T_O$.

\subsubsection{Converting a Tree Transducer into a String Transducer}
Recall that two-way transducers are essentially attributed tree transducers with monadic input and monadic output\footnote{
	Note that the two-way transducers in~\cite{DBLP:conf/lics/FiliotGRS13}
	are defined with 
	a \emph{left end marker} $\vdash$ and a \emph{right end marker} $\dashv$. While the left end marker $\vdash$
	corresponds to the root marker of our tree transducers, the right end marker 
	$\dashv$ has no counterpart. Monadic
	trees can be considered as strings with specific end symbols, i.e. symbols in $\Sigma_0$, 
	that only occur at the end of strings. Thus, $\dashv$ is not required.
	Conversely, two-way  transducers can  test if exactly one end symbol occurs in the input string and if
	it is the rightmost symbol. Hence, two-way transducers
	can simulate tree transducers with monadic input
	and output.
}.
Consider a tree $s\in \text{dom} (\hat{A})$ and let $B$ relabel $s$ into $s'$.
Informally, as on input $s'$, attributes of $A'$ only process nodes occurring on a single path $\rho$ of $s$,
the basic idea is to
`cut off' all nodes from $s'$ not occurring in $\rho$. 
This way, we
effectively
make input trees of $A'$ monadic.

Recall that $A'=(S,I,\Sigma^B,\Delta,a_0,R')$.
Formally, $T_W=(\tilde{S}, \tilde{I},\tilde{\Sigma},\Delta,\tilde{a}, \tilde{R})$, where
$\tilde{S}=S\cup \{\tilde{a}\}$ and $\tilde{a}\notin S$.
We define $\tilde{I}=I\cup I'$ where $I'$ is a set of auxiliary attributes which we define later.
 The set $\tilde{\Sigma}$ is obtained by converting the input alphabet of $A'$ to symbols of rank $1$.
To this end, we first define that the
input alphabet $\tilde{\Sigma}$ of $T_W$ contains all symbols in $\Sigma^B_0$, 
i.e., $\Sigma^B_0\subseteq \tilde{\Sigma}$.
Now consider a symbol $\sigma\in \Sigma^B_k$ with $k>0$. Given such a symbol, we define that the
$\tilde{\Sigma}$ contains the symbols $\langle \sigma,1 \rangle,\dots, \langle \sigma,k \rangle$ of rank~$1$. 
Informally, the idea is that a symbol of the form $\langle \sigma,i \rangle$ 
indicates that the next node is to be
interpreted as the $i$-th child. Thus, trees over  $\tilde{\Sigma}$ are basically encodings of prefixes of
trees over $\Sigma^B$.
For instance,
let $f\in \Sigma^B_2$, $g\in \Sigma^B_1$ and $e\in \Sigma^B_0$ and denote by $\top$ a symbol of rank $0$ not in $\Sigma^B$.
Note that in the following  we omit parentheses  for  monadic trees for better readability.
Then the tree $\langle f,2 \rangle  \langle f,1 \rangle  \langle f,1 \rangle  e$ encodes the prefix
$f(\top, f(f (e, \top) , \top) )$
while the tree
$\langle f,1\rangle \langle g,1\rangle e$ encodes  $f (g (e), \top)$.
The basic idea is that since  attributes of $A'$ only process nodes occurring on a single  path of the input 
tree,
such prefixes are sufficient to simulate $A'$.

In the following, we define the rules of $T_W$.
Due to Lemma~\ref{only}, assume that only rules of $A'$ for $\#$ have ground right-hand sides.
Let $A'$ contains the rule $a(\pi) \rightarrow t\in R'_{\sigma}$ where $\sigma\in \Sigma^B_k$, $k>0$, and $a\in S$. Furthermore, let $\alpha\in S$ such that $\alpha (\pi i)$ with $i\in [k]$ occur in $t$.
Then $T_W$ contains the rule
$
a(\pi) \rightarrow t[\pi i\leftarrow \pi 1] \in \tilde{R}_{\langle \sigma, i\rangle},
$ 
where $[\pi i\leftarrow \pi 1]$ denotes the substitution
that substitutes occurrences of $\pi i$ by~$\pi 1$.
If there are no occurrences of synthesized attributes in $t$ then we define
$a(\pi) \rightarrow t \in \tilde{R}_{\langle \sigma, i\rangle}$. 

Similarly, if $A'$ contains the rule
$b(\pi i) \rightarrow t'\in R'_{\sigma}$  where $b\in I$ and
for some $\alpha\in S$, $\alpha (\pi i)$ occurs in $t'$, then $T_W$ contains the rule
$
b(\pi 1) \rightarrow t'[\pi i\leftarrow \pi 1] \in \tilde{R}_{\langle \sigma,i\rangle }.
$ 
If no synthesized attributes occur in $t'$ then
$b(\pi 1) \rightarrow t' \in \tilde{R}_{\langle \sigma',i\rangle }$. 
We remark that since $\hat{A}$ has the string-like property,
$A'$ will never apply a rule of the form $b(\pi i) \rightarrow t'$ where $\alpha (\pi j)$
with $j\neq i$ occurs in $t'$. Thus, we do not need to consider such rules.

Recall that $\Sigma^B_0= \tilde{\Sigma}_0$. For all $\sigma\in \Sigma^B_0$, we define $R'_{\sigma} \subseteq \tilde{R}_{\sigma}$.
Finally, we define 
$R'_\# \subseteq \tilde{R}_\#$.
Clearly, the rules defined above can be used to simulate~$A'$.

As we have defined that a fresh attribute $\tilde{a}$ as the initial attribute of $T_W$ instead of $a_0$,
the reader might have guessed that we are not finished yet.
For the correctness of subsequent arguments, 
we require a technical detail: 
We require that
 the domain of $T_W$ only consists of trees  $\tilde{s}$
for which a tree $s\in \text{range} (B)$ exists
such that $\tilde{s}$ encodes a prefix of $s$.
In particular, we can only guarantee that a one-way transducer equivalent to
$T_W$ exists if the domain of $T_W$ only consists of such trees.
If $\tilde{s}$ encodes a prefix of $s\in \text{range} (B)$ then
we also say that $\tilde{s}$ \emph{corresponds to}~$s$.

\begin{example}
Consider  the $att^R$ $\hat{A}_2=(B_2, A_2')$ constructed in Example~\ref{att example 4}
and in particular its relabeling $B_2$ constructed in Example~\ref{att example 3}.
Furthermore, consider the trees 
$
\tilde{s}_1=\langle f_{\varrho_1,\varrho_2},1\rangle    d
$ 
and 
$ 
\tilde{s}_2=\langle f_{\varrho_1,\varrho_2},2 \rangle   d.
$
Here, the tree $\tilde{s}_2$ encodes the tree $f_{\varrho_1,\varrho_2} (\top ,d)$
which is a prefix of the output tree $s_2=f_{\varrho_1,\varrho_2} (e,d) \in \text{range} (B_2)$.
Hence, $\tilde{s}_2$ 
corresponds to $s_2$.
The tree $\tilde{s}_1$ however encodes  the tree $f_{\varrho_1,\varrho_2} (d,\top)$.
By definition of the relabeling~$B_2$, there is no output tree in the range of $B_2$ of
which $f_{\varrho_1,\varrho_2} (d,\top)$ is a prefix. 
Specifically, this is because
a node of an output tree of of $B_2$ is only labeled by
$f_{\varrho_1,\varrho_2}$ if its left subtree is $e$.
Hence,
$\tilde{s}_1$ corresponds to no output tree of~$B_2$.
\end{example}

To check whether or not for a given input tree $\tilde{s}$ of $T_W$ an output tree~$s\in \text{range}(B)$ exists such that
$\tilde{s}$ corresponds to $s$, we proceed as follows.
As $B$ is a relabeling, its range is effectively recognizable, i.e.,
a bottom-up tree automaton $\bar{B}$ that accepts precisely the trees in $\text{range} (B)$ exists and
can be constructed.

A \emph{bottom-up tree automaton} is a bottom-up relabeling
where the input and output alphabet are identical and all rules are of
the form 
$
\sigma (p_1 (x_1),\dots, p_k (x_k))\rightarrow p (\sigma (x_1,\dots,x_k)),
$
where $\sigma$ is a symbol and $p,p_1,\dots,p_k$ are states of the automaton.
In the following we allow bottom-up tree automata to be nondeterministic.
The language accepted by a bottom-up tree automaton is its domain.

Given the automaton $\bar{B}$, we construct a  bottom-up tree automaton $\bar{B}'$ that accepts exactly those trees~$\tilde{s}\in T_{\tilde{\Sigma}}$
for which a tree $s\in \text{range}(B)$ exists such that $\tilde{s}$ corresponds to~$s$.
W.l.o.g. assume that for all states $l$ of $\bar{B}$, $\text{dom}_B (l)\neq \emptyset$.
We define that if for $\sigma\in \Sigma^B_k$, the rule
$
\sigma (l_1 (x_1),\dots,l_k  (x_k))\rightarrow l(\sigma (x_1,\dots,x_k))
$
is included in $\bar{B}$ then 
$
\langle \sigma,i\rangle (l_i (x_1)) \rightarrow l(\langle \sigma,i\rangle (x_1))
$
is a rule of $\bar{B}'$. Note that $\bar{B}'$ may be nondeterministic.
We define that  $\bar{B}'$ has the same final states as $\bar{B}$.

Now, let $\tilde{s}$  be the input tree of $T_W$.
Using $\bar{B}'$, we check whether or not a tree $s\in \text{range} (B)$ exists such that
$\tilde{s}$ corresponds to $s$ with the following procedure.
Consider the tree $\tilde{s}^\#$. 
Informally, 
$T_W$ starts by going to the leaf of $\tilde{s}^\#$ 
and subsequently simulating $\bar{B}'$ (without producing any output symbols).
Note that as $\tilde{s}^\#$ is a monadic tree, it has precisely one leaf.
To simulate $\bar{B}'$,  the states of $\bar{B}'$ are essentially considered as inherited attributes.
During the simulation of $\bar{B}'$,  $T_W$ goes  back to the root marker in a bottom-up fashion.
If it reaches the root marker with a final state of $\bar{B}'$; in other words if  a tree $s\in \text{range} (B)$ exists such that
$\tilde{s}$ corresponds to $s$;
then $T_W$ starts to simulate $A'$.

Recall that $\tilde{a}$ is the initial attribute of $T_W$.
First of all, the rules for `going to the leaf of $\tilde{s}^\#$' are defined in 
a straight forward manner:
Let $\sigma$ be  a symbol in $\tilde{\Sigma}$ of rank~$1$.
Then we define that $T_W$ contains the rule
$
\tilde{a}(\pi) \rightarrow \tilde{a} (\pi 1) \in \tilde{R}_{\sigma}
$
To start the simulation of $\bar{B}'$, we introduce the rule
$
\tilde{a}(\pi) \rightarrow l (\pi) \in \tilde{R}_{e}
$
for $e\in \tilde{\Sigma}_0$ if the rule $e\rightarrow l (e)$ occurs in $\bar{B}'$.
The remaining rules to simulate $\bar{B}'$ are defined in a
straight forward manner as well:
Let $\sigma$ be a symbol in $\tilde{\Sigma}$ of rank $1$.
Then a rule $\sigma (l' (x_1)) \rightarrow l(\sigma (x_1))$ of $\bar{B}'$
induces the rule
$
l' (\pi 1) \rightarrow l (\pi) \in \tilde{R}_{\sigma}
$
of $T_W$.

Let $l_0$ be a final state of $\bar{B}'$.
It should be clear that if $T_W$ reaches the root marker with $l_0$, i.e., if
$
\tilde{a} (1) \Rightarrow_{T_W, \tilde{s}^\#}^* l_0 (1),
$ 
then this means that $\tilde{s}$ is accepted by $\bar{B}'$. 
In this case $T_W$ begins to simulate $A'$. To begin the simulation, we define that
$T_W$ also contains the rule
$
l_0 (\pi 1) \rightarrow a_0 (\pi 1) \in \tilde{R}_\#.
$ 
Recall that $a_0$ is the initial attribute of $A'$.
To illustrate our procedure, consider the following example.

\begin{example}
	Consider the $att^R$ $\hat{A}_2=(B_2, A_2')$ constructed in Example~\ref{att example 4}.
	We now convert this $att^R$ into a two-way transducer $T_W$ 
	using the  procedure above.
	To do so, we require that
	the domain of $T_W$ only consists of trees~$\tilde{s}$ for which a tree $s\in \text{range} (B_2)$ exists such that $\tilde{s}$ corresponds to $s$.
	Consider the following bottom-up tree automaton $\bar{B}$ which
	recognizes the range of $B_2$. Recall that the range of a relabeling is effectively recognizable;
	hence an automaton such as $\bar{B}$ can always be obtained.
	In particular, the automaton $\bar{B}$ 
	is obtained from $B_2$ in a straight-forward manner. 
	In addition to the rules $e \rightarrow  \varrho_1 (e)$ and $d \rightarrow  \varrho_2 (d)$,
	$\bar{B}$ also contains the rules 
	\[
	\begin{array}{ccc c ccc }
		f_{\varrho, \varrho_1}(\varrho (x_1), \varrho_1 (x_2)) & \rightarrow & \varrho_4 (f_{\varrho, \varrho_1} (x_1,x_2))  &\ &
		 f_{\varrho, \varrho_2}(\varrho (x_1), \varrho_2 (x_2)) & \rightarrow & \varrho_3 (f_{\varrho, \varrho_2} (x_1,x_2))  \\
		f_{\varrho, \varrho_3}(\varrho (x_1), \varrho_3 (x_2)) & \rightarrow & \varrho_3 (f_{\varrho, \varrho_3} (x_1,x_2)) & \ &
		 f_{\varrho, \varrho_4} (\varrho (x_1), \varrho_4 (x_2)) & \rightarrow & \varrho_4 (f_{\varrho, \varrho_4} (x_1,x_2)),\\
	\end{array}
	\]
	where $\varrho\in \{ \varrho_1,\dots,\varrho_4\}$.
	All states of $\bar{B}$
	are final states.	
	Given $\bar{B}$ we first construct a  bottom-up tree automaton~$\bar{B}'$ recognizing the set of all trees which correspond to some tree in $\text{range} (B_2)$.
	 For $\varrho\in \{ \varrho_1,\dots,\varrho_4\}$, it contains the rules
	\[
	\begin{array}{ccc c ccc}
		\langle f_{\varrho, \varrho_1},1 \rangle (\varrho (x_1)) & \rightarrow & \varrho_4 ( \langle f_{\varrho, \varrho_1}, 1 \rangle (x_1))  &\ &
		 \langle f_{\varrho, \varrho_2}, 1 \rangle (\varrho (x_1)) & \rightarrow & \varrho_3 ( \langle f_{\varrho, \varrho_2},1 \rangle (x_1))  \\
		\langle f_{\varrho, \varrho_3},1 \rangle (\varrho (x_1)) & \rightarrow & \varrho_3 (\langle f_{\varrho, \varrho_3},1\rangle  (x_1)) 
		&\ &
		 \langle f_{\varrho, \varrho_4},1\rangle  (\varrho (x_1)) & \rightarrow & \varrho_4 (\langle f_{\varrho, \varrho_4},1\rangle (x_1))\\
		\langle f_{\varrho, \varrho_1},2 \rangle (\varrho_1 (x_1)) & \rightarrow & \varrho_4 ( \langle f_{\varrho, \varrho_1}, 2 \rangle (x_1))  &\ &
		 \langle f_{\varrho, \varrho_2}, 2 \rangle (\varrho_2 (x_1)) & \rightarrow & \varrho_3 ( \langle f_{\varrho, \varrho_2},1 \rangle (x_1))  \\
		\langle f_{\varrho, \varrho_3},2 \rangle (\varrho_3 (x_1)) & \rightarrow & \varrho_3 (\langle f_{\varrho, \varrho_3},2\rangle  (x_1)) &\ &
		 \langle f_{\varrho, \varrho_4},2\rangle  (\varrho_4 (x_1)) & \rightarrow & \varrho_4 (\langle f_{\varrho, \varrho_4},1\rangle (x_1)),\\
	\end{array}
	\]
	as well as  $e \rightarrow  \varrho_1 (e)$ and $d \rightarrow  \varrho_2 (d)$.
	As with $\bar{B}$, all states of $\bar{B}'$ are final states.
	
	To ensure that the domain of $T_W$ fits our requirement, $T_W$ first simulates~$\bar{B}'$.
	To start with, we define that the fresh attribute $\tilde{a}$ is the initial attribute of
	$T_W$.
	Denote $\tilde{\Sigma}$ the input alphabet of $T_W$. 
	In the following,  $\varrho_1,\dots,\varrho_4$ are considered inherited attributes.
	Let $\sigma\in \tilde{\Sigma}_1$. Then we define
	that $T_W$ contains the rule
	$\tilde{a} (\pi) \rightarrow \tilde{a} (\pi 1)\in \tilde{R}_{\sigma}$.
	Given $\bar{B}'$ and $e,d\in \tilde{\Sigma}_0$, we obtain the rules
	$\tilde{a}(\pi) \rightarrow \varrho_1 (\pi) \in \tilde{R}_e$ and 
	$\tilde{a}(\pi) \rightarrow \varrho_2 (\pi) \in \tilde{R}_d$.
	Furthermore, we obtain the rules
	\[
	\begin{array}{ccc}
		\varrho(\pi 1) \rightarrow \varrho_4 (\pi) \in \tilde{R}_{\langle f_{\varrho, \varrho_1}, 1 \rangle} & \quad  & \varrho(\pi 1) \rightarrow \varrho_3 (\pi) \in \tilde{R}_{\langle f_{\varrho, \varrho_2}, 1 \rangle}  \\
		\varrho(\pi 1) \rightarrow \varrho_3 (\pi) \in \tilde{R}_{\langle f_{\varrho, \varrho_3}, 1 \rangle} & \quad & \varrho(\pi 1) \rightarrow \varrho_4 (\pi) \in \tilde{R}_{\langle f_{\varrho, \varrho_4}, 1 \rangle}\\
		\varrho_1(\pi 1) \rightarrow \varrho_4 (\pi) \in \tilde{R}_{\langle f_{\varrho, \varrho_1}, 2 \rangle} & \quad  & \varrho_2(\pi 1) \rightarrow \varrho_3 (\pi) \in \tilde{R}_{\langle f_{\varrho, \varrho_2}, 2 \rangle}  \\
		\varrho_3(\pi 1) \rightarrow \varrho_3 (\pi) \in \tilde{R}_{\langle f_{\varrho, \varrho_3}, 2 \rangle} & \quad & \varrho_4(\pi 1) \rightarrow \varrho_4 (\pi) \in \tilde{R}_{\langle f_{\varrho, \varrho_4}, 2 \rangle},\\
	\end{array}
	\]
		where $\varrho\in \{ \varrho_1,\dots,\varrho_4\}$ from $\bar{B}'$.
	It should be clear that if for a tree $\tilde{s}$ over $\tilde{\Sigma}$, it holds that
	$
	\tilde{a} (1) \Rightarrow_{T_W, \tilde{s}^\# }^* \varrho (1), 
	$
	 then  $\tilde{s}$ corresponds to some tree $s\in \text{range} (B_2)$. 
	 In this case, we can proceed to simulate $A_2'$.
	To this end, we define the rule
	$\varrho (\pi 1) \rightarrow a (\pi 1)\in \tilde{R}_\#$.
	Recall that $a$ is the initial attribute of the $att$ $A_2'$.

	We now specify the remaining rules, i.e., the rules with which $T_W$ simulates~$A_2'$.
	Let $\varrho,\varrho'\in \{\varrho_1,\dots,\varrho_4\}$.
	Then the rule set $\tilde{R}_{\langle f_{\varrho,\varrho'},1\rangle}$ is 
	obtained from the rule set $R_{ f_{\varrho,\varrho'} }$ of $A_2'$ by removing all rules 
	where $\pi 2$ occurs either on the left or right-hand side.
	Analogously, the rule set	$\tilde{R}_{\langle f_{\varrho,\varrho'},2\rangle}$
	is obtained from $R_{ f_{\varrho,\varrho'} }$ of $A_2'$ by removing all rules 
	where $\pi 1$ occurs either on the left or right-hand side.
	Rules for the root marker as well as rules for symbols of rank $0$
	are taken over from~$A_2'$.
\end{example}

\noindent 
By construction of $T_W$, it is clear that
 the following  holds.

\begin{lemma}\label{lemma 7}
	Consider the $att^R$ $\hat{A}= (B,A')$ and the two-way transducer $T_W$ constructed from $\hat{A}$.
	Let $\tilde{s}$ be a tree over  $\tilde{\Sigma}$. 
	If on input $\tilde{s}$, $T_W$ outputs  $t$ then for all $s\in \text{range} (B)$ such
	that $\tilde{s}$ corresponds to~$s$,
	$A'$ also produces  $t$ on input $s$.
\end{lemma}

\subsubsection{From String Transducers back to Tree Transducers}
In the following, consider the two-way transducer  $T_W$. Assume that the procedure 
of~\cite{DBLP:conf/lics/FiliotGRS13} yields a one-way transducer $T_O$ that
is equivalent to $T_W$.
Recall that a one-way transducer is in essence a top-down tree transducer 
with monadic input and monadic output.

Given the one-way transducer $T_O=(\bar{S},\bar{I},\tilde{\Sigma},\Delta, \bar{a}_0, \bar{R})$, we now construct a top-down transducer
$T'=(\bar{S},\bar{I},\Sigma^B,\Delta, \bar{a}_0, \acute{R})$ that produces output trees on the range of $B$.
To do so, $\acute{R}$ is constructed as follows: 
Let $q(\langle \sigma,i\rangle (x_1)) \rightarrow t \in \bar{R}$ 
where $\sigma\in \Sigma^B_k$ and $i\in  [k]$. This rule
induces the rule $q(\sigma (x_1,\dots,x_k)) \rightarrow \hat{t}\in \acute{R}$ where
$\hat{t}$ is obtained from $t$ by substituting occurrences of $x_1$ by $x_i$, e.g., if
$t= f(g (q' (x_1)))$ then $\hat{t}= f(g (q' (x_i)))$.

Recall that the domain of $T_W$  only consists  of  trees  $\tilde{s}\in T_{\tilde{\Sigma}}$ for which
$s\in \text{range} (B)$ exists such that $\tilde{s}$ corresponds to $s$.
As $T_W$ and $T_O$ are equivalent, the domain
of $T_O$ also consists of such trees.
Hence, by construction,  the following holds.

\begin{lemma}\label{lemma 8}
	Consider the top-down transducer $T'$ constructed from the one-way transducer $T_O$.
	Let $\tilde{s}$ be a tree over  $\tilde{\Sigma}$. 
	If on input $\tilde{s}$, $T_O$ outputs  $t$ then for all $s\in \text{range} (B)$ such
	that $\tilde{s}$ corresponds to~$s$,
	$T'$ also produces  $t$ on input $s$.
\end{lemma}
With Lemmas~\ref{lemma 7} and~\ref{lemma 8}, it can be shown that the following holds.
\begin{lemma}\label{lemma 9}
	The top-down transducer $T'$ and the $att$ $A'$ are equivalent on the range of $B$.
\end{lemma}
\begin{proof}
	Let $s\in \text{range} (B)$. Let $A'$ produce the tree $t$ on input $s$.
	Since the string-like property holds and by construction of $T_W$, it follows that a tree $\tilde{s}$
	over $\tilde{\Sigma}$ exists such that $\tilde{s}$ corresponds to $s$
	and $T_W$ produces $t$ on input $\tilde{s}$.
	Since $T_W$ and $T_O$ are equivalent, $T_O$  also produces $t$ on input $\tilde{s}$.
	Due to Lemma~\ref{lemma 8}, it follows that $T'$ produces $t$ on input $s$ as well.
	
	The converse direction follows analogously with Lemma~\ref{lemma 7}.
Note that by construction of $T'$, on input $s$,
 the states of $T'$ can only process nodes occurring on a single path $\rho$
	of $s$. Furthermore, by construction of $T'$ it is implied that if
	 $T'$ produces $t$ on input $s$ then some tree $\tilde{s}$ 
	 over $\tilde{\Sigma}$ exists such that
	 $\tilde{s}$ corresponds to $s$ and
	  $T_O$ produces $t$ on input $\tilde{s}$.
\end{proof}

Due to Lemma~\ref{lemma 9}, it follows that $\hat{A}=(B,A')$ and $N=(B,T')$ are equivalent.
We remark that there is still a technical detail left.
Recall that our aim is to construct a $dt^R$ $T$ equivalent to $\hat{A}$.
However, the procedure 
of~\cite{DBLP:conf/lics/FiliotGRS13} may yield a functional, nondeterministic one-way transducer $T_O$.
Therefore, $T'$ and hence $N$ may be nondeterministic but functional.
As shown in~\cite{DBLP:journals/ipl/Engelfriet78}, we can easily compute  a $dt^R$ equivalent to $N$, thus obtaining a $dt^R$ equivalent to $\hat{A}$.
In summary, our procedure above yields the following.

\begin{lemma}\label{lemma aux}
 If a one-way transducer equivalent to $T_W$ exists
then a $dt^R$ equivalent to the $att^R$ $\hat{A}$
exists and can be constructed.
\end{lemma}
 
What about the inverse implication?
Furthermore, note that the arguments presented above are based on the assumption that
the procedure 
of~\cite{DBLP:conf/lics/FiliotGRS13} yields a (possibly nondeterministic) one-way transducer equivalent to $T_W$. 
Now the question is, does such a one-way transducer
always  exists if a $dt^R$ equivalent to $\hat{A}$ exists? The answer to this question is indeed affirmative.
In particular a one-way transducer equivalent to $T_W$ exists due to the following lemma.

\begin{lemma}\label{lemma 10}
	If a $dt^R$ $T$ equivalent to $\hat{A}$ exists, then a	
	(nondeterministic) $t^R$
	$N=(B, N')$  exists such that
	$\hat{A}$ and $N$ are equivalent.
\end{lemma}

Before we prove Lemma~\ref{lemma 10}, note that the $att^R$ $\hat{A}$ and the $dt^R$ $T$
equivalent to $\hat{A}$ in Lemma~\ref{lemma 10}
may not necessarily use the same bottom-up relabeling. In fact, it may be possible that no
$dt^R$ exists which is equivalent to $\hat{A}$ and uses the same relabeling.
However, 
the nondeterministic $t^R$ $N$ does use the same bottom-up relabeling
as $\hat{A}$. 
We will later construct the one-way transducer $T_O$ from $N$ using this exact property. 
We now prove Lemma~\ref{lemma 10}.

\begin{proof}
	By~\cite{DBLP:journals/ipl/FulopM00},
	$\text{dom} (A')$ is effectively regular, i.e., a deterministic bottom-up tree automaton recognizing
	$\text{dom} (A')$ can be constructed. 
	Thus we can assume that $\text{range} (B)\subseteq \text{dom} (A')$, which implies $\text{dom} (B) =\text{dom} (\hat{A})$.
	 In other words, trees not in $\text{dom} (\hat{A})$
	are filtered by $B$.
	
\underline{\textit{Main Idea}}.
Before we begin our proof, we briefly sketch the main idea.
First recall that 
a node $v$ labeled by 
$\sigma$ is relabeled by $B$ into
$\sigma_{\varrho_1,\dots,\varrho_k}$
if for $i\in [k]$, 
the $i$-th subtree of $v$ is a tree in
$\text{dom}_{B} (\varrho_i)$. 
In the following, we denote  for each state $\varrho$ of $B$
by
$s_{\varrho}$ an arbitrary but fixed tree in $\text{dom}_{B} (\varrho)$.

By our premise, a $dt^R$ $T=(B_T, T')$ equivalent
to $\hat{A}$ exists. 
Without loss of generality assume that the bottom-up relabeling $B_T$ of $T$
operates analogously to $B$, i.e., 
a node $v$  labeled by $\sigma$ is relabeled by $B_T$
into $\sigma_{l_1,\dots,l_k}$ 
if for $i\in [k]$, 
the $i$-th subtree of $v$ is a tree in
$\text{dom}_{B_T} (l_i)$.

We show that $N$  can simulate $T$ using its bottom-up relabeling $B$ and
the following property which we call the \emph{substitute-property}.
Let $s\in \text{dom} (\hat{A})$ and let $B$ relabel $s$ into $\hat{s}$.
Let $v_1$ and $v_2$ be nodes of $\hat{s}$ with the same parent.
Since $\hat{A}$ has the string-like property,
on input $\hat{s}$, either $v_1$ or $v_2$ is not processed by attributes of $A'$.
Assume that $v_1$ is not processed and that $s/v_1 \in \text{dom}_B (\varrho)$.
Then 
\[
\tau_{\hat{A}} (s)=\tau_{\hat{A}} (s[v_1\leftarrow s_\varrho])
\]
holds.
Informally, this means that $s/v_1$ can be substituted by $s_\varrho$ without affecting the output of
the translation.
Since $\hat{A}$ and $T$ are  equivalent by our premise,
$\tau_{T} (s) = \tau_{T} (s[v_1\leftarrow s_\varrho])$ follows, i.e.,
$s/v_1$ can be substituted by $s_\varrho$ in a translation of $T$ without affecting the output
as well.

We now sketch how the $dt^R$ $T$ is simulated by $N$.
Let $s\in \text{dom} (\hat{A})$ and let $B$ relabel $s$ into $\hat{s}$. 
Let $v$ be a node and let $s/v= \sigma (s_1,\dots, s_k)$ and
$\hat{s}/v= \sigma_{\varrho_1,\dots,\varrho_k} (\hat{s}_1,\dots,\hat{s}_k)$.
Furthermore, let $\hat{q}$ be a state of $N'$. 
Let  $\hat{q}$  processes the node $v$ on input $\hat{s}$.
The state $\hat{q}$ is associated with a state $q$ of $T'$ along with a state $l$ of $B_T$.

Unsurprisingly, the main difficulty in simulating the $dt^R$ $T$ is that since the bottom-up relabeling of $N$ 
is $B$ and not $B_T$,  $N$ does not know how the node $v$ of~$s$ would be relabeled by $B_T$.
Using nondeterminism, the obvious approach is that $N$ simply guesses how $v$
 could have been  relabeled by $B_T$. By definition, how the node $v$ is relabeled on input
 $s$ by $B_T$  depends on its subtrees $s_1,\dots,s_k$, i.e.,
 for each $i\in [k]$, we need to guess the state $l_i$ of $B_T$ such that $s_i\in \text{dom}_{B_T} (l_i)$.
 Obviously, all such guesses must be checked for correctness.
 To do so, $N'$
  must read all subtrees $\hat{s}_1,\dots,\hat{s}_k$. However, since $N'$ is a top-down transducer
  with monadic output, 
  $N'$ can read at most one of the subtrees $\hat{s}_1,\dots,\hat{s}_k$.
This is where the substitute-property comes into play.
Using $B$ and the substitute-property,  $N'$ proceeds as follows:
First $N'$ determines which child node of $v$ is processed by attributes of $A'$
on input $\hat{s}$ and which are not.
Assume that on input $\hat{s}$, attributes of $A'$ process the nodes $v.1$.
Consequently, the nodes $v.2,\dots, v.k$ are not processed by any attribute of $A'$
due to the string-like property.
Note since $v$ is labeled by 
$\sigma_{\varrho_1,\dots,\varrho_k}$,
for $i\in [k]$, $s_i\in \text{dom}_B (\varrho_i)$.
Due to the substitute-property,  $s_i$
may be replaced by $s_{\varrho_i}$ for $i\neq 1$ without affecting the produced output tree.
Thus,
$N$  acts as if the $i$-th subtree of $v$ 
was  $s_{\varrho_i}$ for $i\neq 1$ and behaves accordingly. 
Note that since they are fixed, for all trees $s_\varrho$, where $\varrho$ is a state of $B$,
the state $l_\varrho$ such that
$s_\varrho\in \text{dom}_{B_T} (l_\varrho)$
can be precomputed.
In particular, when processing $v$, $N'$ guesses a state $l'$ such that 
\[
\sigma (l' (x_1), l_{\varrho_2}\dots,l_{\varrho_k} (x_k))\rightarrow l (\sigma_{l', l_{\varrho_2}\dots,l_{\varrho_k}} (x_1,\dots,x_k))
\]
is
a rule of $B_T$. With this guess,  the state $\hat{q}$ then `behaves' as the state $q$ of
$T'$ would when processing a node labeled by
$\sigma_{l', l_{\varrho_2}\dots,l_{\varrho_k}}$. 
Afterwards, the subtree $\hat{s}_1$ is read by $N'$ to check whether or not guessing $l'$ is correct.

\underline{\textit{Construction of ${N}'$}}.
Recall that $\Sigma^B$ is the output alphabet of $B$
We define $N'=(\acute{S}, \emptyset, \Sigma^B, \Delta, \hat{q}_0,\acute{R})$.
In addition to the initial state $\hat{q}_0$, $\acute{S}$ consists of 
auxiliary states, which we specify later, and states of the form
$(q',l',a,\gamma)$, where $q'$ and $l'$ are  states  
of $T'$ and $B_T$, respectively,  $a$ is a synthesized attribute of $A'$
and $\gamma\subseteq I\times S$.

Consider   a tree $\hat{s}$ such that $(s,\hat{s})\in \tau_B$ for some
$s\in T_\Sigma$.
Recall that for such a tree $\hat{s}$, $N'$ needs to determine which nodes of $\hat{s}$ are processed by attributes
of $A'$ on input $\hat{s}$ and which are not.
In the following, we describe how $N'$ does so using states of the form $(q',l',a,\gamma)$.

\underline{\textit{Determining the nodes of $\hat{s}$  processed by attributes
of $A'$.}} 
Assume that a state of $N'$ of the form $(q',l',a,\gamma)$ processes the node $v$ on input~$\hat{s}$, i.e.,
assume  that $t$ exists such that
$\hat{q}_0 (1) \Rightarrow_{N', \hat{s}^\#}^* t$ and $(q',l',a,\gamma) (1.v)$ occurs in $t$.
This is to be interpreted as follows:
It means that in a translation of $A'$ on input $\hat{s}$,
the node $v$ is processed by attributes of $A'$.
In particular, the attribute $a$ is the first attribute to process $v$. 
In other words, trees $t_1\dots,t_n$ exist such that
\[
(a_0,1)\Rightarrow_{A',\hat{s}^\#} t_1 \Rightarrow_{A',\hat{s}^\#} t_2 \Rightarrow_{A',\hat{s}^\#} \cdots \Rightarrow_{A',\hat{s}^\#} t_n
\]
such that $a (1.v)$ occurs in $t_n$ and
for $i<n$ it holds that if
$\hat{\alpha} (\nu) \in \text{SI} (\hat{s}^\#)$ occurs in the tree $t_i$,
then  $\nu$ is a proper ancestor of~$1.v$.
Recall that due to string-like property, on input $\hat{s}$, only nodes on a single path of $\hat{s}$
are processed by attributes of~$A'$.
Thus, before $v$ is processed by $a$, only ancestors of $v$ are processed by attributes of $A'$.
We remark that since output trees are monadic,
at most one node of the input tree is processed by an attribute of $A'$ at any given time.
Hence, for all nodes there exists a unique attribute which is the first to process that node or
that node is not processed by any attribute at all.

Consider the component $\gamma$ of the state $(q',l',a,\gamma)$.
If $(b,a)\in \gamma$ then this means that trees $\acute{t_1},\dots,\acute{t}_n \in T_\Delta [\text{SI} (\hat{s}^\#)]$
exist such that
\begin{enumerate}
	\item  $a(1.v)$ occurs in the tree $\acute{t}_n$
	\item $(b,1.v)\Rightarrow_{A', \hat{s}^\#} \acute{t}_1 \Rightarrow_{A', \hat{s}^\#}  \acute{t}_2 \Rightarrow_{A', \hat{s}^\#}\dots \Rightarrow_{A', \hat{s}^\#} 
	\acute{t}_m $ holds and
	\item for $i<m$, if
	$\hat{\alpha} (\nu) \in \text{SI} (\hat{s}^\#)$ occurs in  $\acute{t}_i$,
	then  $\nu$ is a proper ancestor of~$1.v$.
\end{enumerate}

Assume that the node $v$ has $k$ child nodes.
Since  $a$ is the first attribute of $A'$ to process $v$
and given the definition of $\gamma$,
it should be clear that using only $a$ and~$\gamma$
along with 
 rules of $A'$ for  $\hat{s} [v]$, the transducer $N'$ can compute
which
child nodes of $v$ are processed by attributes of $A'$ and which are not.
In particular, the first attribute to process a child node of $v$ can be computed.

\underline{\textit{Defining the rules of ${N}'$}}.
Subsequently, we define the rules of $N'$, beginning with the rules for its initial state $\hat{q}_0$.
Let $q_0$ be the initial state of $T'$. Recall that $a_0$ is the initial attribute of $A'$.
Furthermore, denote by $\gamma'$ the set
\begin{multline*}
	\gamma'= \{ (b,a) \in I\times S \mid A' \text{ contains a rule of the form }b(\pi 1) \rightarrow \psi \in R_\#\\
	\text{ such that } a(\pi 1)\text{ occurs in } \psi\}.
\end{multline*}
Let $l'$ be a final state of $B_T$.
To compute the rules of $N'$ for $\hat{q}_0$, we basically have to compute the rules
for the state $(q_0, l',a_0, \gamma')$.
In particular, 
for $\sigma_{\varrho_1,\dots,\varrho_k}\in \Sigma B$, $N'$ contains the rule 
$\hat{q}_0 (\sigma_{\varrho_1,\dots,\varrho_k} (x_1))\rightarrow \xi'$ 
if and only if it also contains
$(q_0, l',a_0, \gamma') (\sigma_{\varrho_1,\dots,\varrho_k} (x_1))\rightarrow \xi'$.

To define the rules for  $(q_0, l',a_0, \gamma')$
consider the following.
Let $(q,l,a,\gamma)$ be an arbitrary state of $N'$.
For  $(q,l,a,\gamma)$
and  $\sigma_{\varrho_1,\dots,\varrho_k}$ we define the following rules.

\underline{\textit{Case 1:}}
Assume that using $a$ and $\gamma$,
$N'$ computes that no child node of the current input node
is processed by attributes  of $A'$.
In this case, $N'$ assumes that the $i$-th
subtree of the current input node prior to the relabeling by $B$ has been $s_{\varrho_i}$ for $i\in [k]$
due to the substitute-property.   
Hence, consider the tree  $s=\sigma(s_{\varrho_1},\dots,s_{\varrho_k})$. Denote by
$\tilde{s}$ the tree obtained from $s$ via the relabeling $B_T$ of $N$.
Let $t$ be the tree produced by $q$ on input $\tilde{s}$, i.e., let $t=\text{nf} (\Rightarrow_{T', \tilde{s}}, q (\epsilon))$.
If $s\in\text{dom}_{B_T} (l)$  and $t\in T_\Delta$ then for $N'$, we define  the rule
\[
(q,l,a,\gamma)(\sigma_{\varrho_1,\dots, \varrho_k}(x_1,\dots,x_k))\rightarrow t.
\]

\underline{\textit{Case 2:}} Assume that using $a$ and $\gamma$,
 $N'$ computes that the $i$-th child node of the current input node
 is processed by attributes of $A'$ where $i\in [k]$. Consequently,
no attributes of $A'$ process any of the remaining child nodes of the current input node.
Therefore,  $N'$ assumes for $j\neq i$ that the $j$-th
subtree of the current input node prior to the relabeling by $B$ has been $s_{\varrho_j}$ 
due to the substitute-property. 
Let $s_{\varrho_j} \in \text{dom}_{B_T} (l_j)$
for $j\neq i$.
For all states $l_i$ of $B_T$ such that
\begin{enumerate}
	\item  $\sigma (l_1(x_1),\dots,l_k (x_k))\rightarrow l (\sigma_{l_1,\dots, l_k} (x_1,\dots,x_k))$ is a rule of $B_T$  and 
	\item    the right-hand side $t$ 
	of the rule of $T'$ for $q$ and $\sigma_{l_1,\dots, l_k}$ 
	contains an occurrence of $q' (x_i)$, where $q'$ is a state of $T'$
\end{enumerate}
we define a rule
\[
(q,l,a,\gamma)(\sigma_{\varrho_1,\dots, \varrho_k}(x_1,\dots,x_k))\rightarrow t'
\]
for $N'$,
where $t'$ is obtained from $t$ by substituting the occurrence of $q' (x_i)$
by $(q',l_i, a_i, \gamma_i) (x_i)$.
Here, $a_i$ denotes the first attribute to process the $i$-th child of the current input node.
The attribute $a_i$ along with the set $\gamma_i \subseteq I\times S$ can be computed from $a$ and $\gamma$
in conjunction with the rules of $A'$ for $\sigma_{\varrho_1,\dots,\varrho_k}$
in a straightforward manner.

Note the rule defined above
 guesses the state $l_i$ of $B_T$ for the $i$-th subtree of the current input node.
 This rule also ensures that the the $i$-th subtree is read, meaning that
the guess can be checked.
Clearly this requires that
in the
rule of $T'$
for $q$ and $\sigma_{l_1,\dots, l_k}$,
$q' (x_i)$ occurs in $t'$.
We now consider the case in which such a rule is not available.
This is also the most complicated case.

\underline{\textit{Case 3:}}
Assume that using $a$ and $\gamma$,
$N'$ computes that the $i$-th child node of the current input node
is processed by attributes  of $A'$ where $i\in [k]$. 
Thus, as in the previous cases, due to the substitute-property, $N'$ assumes that the $j$-th
subtree of the current input node prior to the relabeling by $B$ has been $s_{\varrho_j}$ for $j\neq i$. 
Let $s_{\varrho_j} \in \text{dom}_{B_T} (l_j)$
for $j\neq i$.

In the following, we require auxiliary states that are of the form
$(e',l',a',\gamma')$ where 
$e'$ is a symbol of $\Delta$ of rank $0$,
$l'$ is a state  
of $B_T$,  $a$ is a synthesized attribute of $A'$
and $\gamma\subseteq I\times S$.
With these auxiliary states, we now consider the case where a state $l_i$ of $B_T$ exists such that
\begin{enumerate}
	\item  $\sigma (l_1(x_1),\dots,l_k (x_k))\rightarrow l' (\sigma_{l_1,\dots, l_k} (x_1,\dots,x_k))$ is a rule of $B_T$  and 
	\item  the right-hand side 
	of the rule of $T'$ for $q$ and $\sigma_{l_1,\dots, l_k}$ 
	contains an occurrence of $q' (x_\iota)$ where $q'$ is a state of $T'$ and $\iota\neq i$.
\end{enumerate}
Informally, this is the case where $T'$ and $A'$ diverge, that is, where $T'$ and $A'$ process different child nodes of the 
current input node.
In this case, we proceed as follows:
Consider the tree $s=\sigma(\grave{s}_1,\dots, \grave{s}_k)$
where $\grave{s}_i$ is an arbitrary tree in $\text{dom}_{B_T} (l_i)$ and for $j\neq i$, $\grave{s}_j = s_{\varrho_j}$. Let $\tilde{s}$ be obtained from $s$ via the relabeling $B_T$ and let $t\in T_\Delta$
such that $t=\text{nf} (\Rightarrow_{T', \tilde{s}}, q (\epsilon))$.
Since $t\in T_\Delta$ (and thus $t$ is monadic) exactly one symbol occurring in $t$ is of rank $0$.
Let $e$ be this symbol. Let $t'$ is obtained from $t$ by substituting the occurrence of $e$
by $(e,l_i,a_i,\gamma_i)$ where the attribute $a_i$ and $\gamma_i\subseteq I\times S$
are computed from the components $a$ and $\gamma$ of $(q,l,a,\gamma)$
in conjunction with the rules of $A'$ for $\sigma_{\varrho_1,\dots,\varrho_k}$ as in the previous case.
Then we define the rule
\[
(q,l,a,\gamma)(\sigma_{\varrho_1,\dots, \varrho_k}(x_1,\dots,x_k))\rightarrow t'
\]
for $N'$.
Note that the reason why 
 $\grave{s}_i$ is chosen as an arbitrary tree in $\text{dom}_{B_T} (l_i)$
is because $t$ is computed using the rule of $T'$ for $q$ and $\sigma_{l_1,\dots, l_k}$
and
$q' (x_j)$ occurs in the right-hand side 
of that rule. Thus, no matter how $\grave{s}_i$ is chosen,
the produced tree will always be $t$.

A state of the form $(e',l',a',\gamma')$, where 
$e'$ is a symbol of $\Delta$ of rank $0$, simply tests whether the guess of $N'$ is correct, i.e.,
whether or not the $i$-th subtree of the current input node is indeed an element of $\text{dom}_{B_T} (l_i)$. 
If so then it eventually outputs $e'$.
For such a state
$(e',l',a',\gamma')$  we define the following rules.
Assume that the state $(e',l',a',\gamma')$ processes  a node labeled by $\sigma_{\varrho_1,\dots,\varrho_k}$ where $k\geq 0$.
Consider the following two cases:

(a) Assume that 
	using $a'$ and $\gamma'$,
	$N'$ computes that no child nodes of the current input node
	are processed by attributes  of $A'$.
	Then we define the rule
	\[
	(e',l',a',\gamma') (\sigma_{\varrho_1,\dots, \varrho_k}(x_1,\dots,x_k))\rightarrow e'
	\]
	for $N'$
	if   $\sigma(s_{\varrho_1},\dots,s_{\varrho_k})\in \text{dom}_{B_T} (l')$.
	
	(b) Assume that $N'$ computes that the $i$-th child node of the current input node
	is processed by attributes  of $A'$ where $i\in [k]$.
	Due to the substitute property, $N'$ assumes that the $j$-th
	subtree of the current input node prior to the relabeling by $B$ has been $s_{\varrho_j}$ for $j\neq i$.
	Let $s_{\varrho_j} \in \text{dom}_{B_T} (l_j)$
	for $j\neq i$. 
For all states $l_i$ of $B_T$ such that
 $\sigma (l_1(x_1),\dots,l_k (x_k))\rightarrow l' (\sigma_{l_1,\dots, l_k} (x_1,\dots,x_k))$ is a rule of $B_T$, 
we define a rule
\[
(e',l',a',\gamma')(\sigma_{\varrho_1,\dots, \varrho_k}(x_1,\dots,x_k))\rightarrow (e',l_i,a_i,\gamma_i),
\]
where $a_i$ and $\gamma_i$ are obtained as before.

\underline{\textit{Correctness of ${N}$}}.
Since $T$ and $\hat{A}$ are equivalent, it is sufficient to show for our lemma that
$(s,t) \in \tau_N$ if and only if  $(s,t) \in \tau_T$.

First consider the following, Let $\hat{s}$ be the tree obtained from $s$ via the 
bottom-up relabeling $B$. 
Consider $\hat{A}=(B,A')$.
Due to the string-like property,
on input $\hat{s}$, attributes of $A'$ only process nodes occurring on a single path of
$\hat{s}$.  Denote by $V_{A'} (\hat{s})$ the set of all these nodes.
Accordingly, 
denote by $\neg V_{A'} (\hat{s})$ the set of all nodes of $\hat{s}$ that
\emph{not} processed by attributes of $A'$ on input $\hat{s}$.
Let $\neg V_{A'}^a (\hat{s})$ be the set of all nodes $v\in \neg V_{A'}(\hat{s})$
for which no ancestor $v'$ of $v$ exists such that $v' \in  \neg V_{A'}(\hat{s})$.
Note that $V(s)=V(\Hat{s})$.
Consider the tree
\[
\textit{\c{s}}=s[v\leftarrow s_\varrho \mid v \in \neg V_{A'}^a  (\hat{s}) \text{ and } s/v\in \text{dom}_B (\varrho) ].
\]
Due to the substitute property $(s,t) \in \tau_T$ if and only if $(\textit{\c{s}},t) \in \tau_T$.
Therefore, it is sufficient ot show that  $(s,t) \in \tau_N$ if and only if  $(\textit{\c{s}},t) \in \tau_T$.
Due to the definition of the rules of $N'$, the latter follows by straight-forward induction.
\end{proof}

Note that  the $t^R$ 
$N=(B, N')$ constructed in the proof of Lemma~\ref{lemma 10} has the following property:
On input $\hat{s}\in \text{range} (B)$ 
an attribute of $A'$ processes the node~$v$ if and only if  
a state of $N'$ processes~$v$ on input $\hat{s}$.
The existence of  a $t^R$ $N$ 
with this properties
implies the existence of a one-way transducer $T_O$ equivalent to $T_W$. In fact, $T_O$ is obtainable
from $N$ similarly to how $T_W$ is obtainable from $\hat{A}$.

Given, $N=(B, N')$, the transducer $T_O$ is 
the product of a 
top-down transducer $\mathcal{T}$ and a top-down tree automaton $\mathcal{T}'$, i.e.,
$T_O$ is obtained by running $\mathcal{T}$ and $\mathcal{T}'$ in parallel.
Recall that $\tilde{\Sigma}$ is the input alphabet of $T_W$.
We define that $\tilde{\Sigma}$ is also the input alphabet of both  $\mathcal{T}$ and  $\mathcal{T}'$.
Informally, the idea is that the transducer $\mathcal{T}$ is tasked with simulating $N'$ while
the purpose of the automaton $\mathcal{T}'$ is to check whether or not
for an input tree $\tilde{s}$ of $\mathcal{T}$  an output tree $\hat{s}\in \text{range} (B)$ exists such that $\tilde{s}$ corresponds to~$\hat{s}$.

Recall  that  we have previously constructed a  bottom-up tree automaton $\bar{B}'$ that accepts exactly those trees $\tilde{s}$
for which a tree $s\in \text{range} (B)$ exists such that $\tilde{s}$ corresponds to $s$.
Then obviously an equivalent nondeterministic top-down tree automaton can be constructed.
We define that $\mathcal{T}'$ is such an automaton.

Recall that $N'=(\acute{S}, \emptyset, \Sigma^B, \Delta, \hat{q}_0,\acute{R})$.
We define $\mathcal{T}=(\acute{S}, \emptyset, \tilde{\Sigma}, \Delta, \hat{q}_0, \grave{R})$ where
 $\grave{R}$ is defined as follows: Let $\acute{R}$ contain  the rule 
$
q(\sigma' (x_1,\dots,x_k))\rightarrow t
$, where $k>0$.
Let $\hat{q} (x_i)$ where $i\in [k]$ and $\hat{q} \in \acute{S}$ occurs in $t$. Then we define the rule
$
q( \langle \sigma', i \rangle (x_1))\rightarrow \hat{t}
$
for $\mathcal{T}$, where $\hat{t}$ is obtained from $t$ by substituting  $\hat{q} (x_i)$ by  $\hat{q} (x_1)$.
If $t$ is ground then we define the rule
$
q(\langle \sigma', i \rangle (x_1))\rightarrow t
$
for all $i\in [k]$ instead. 
Additionally, $\grave{R}$ contains all rules of $\acute{R}$ where symbols of rank $0$ occur on the left-hand
side.

Before we can run $\mathcal{T}$ and $\mathcal{T}'$ in parallel, there is a
 technical detail left. Note that $\mathcal{T}'$ obviously needs to read the whole input tree $\tilde{s}$  to decide whether or not for  $\tilde{s}$ an output tree $\hat{s}$ of $B$ exists such that $\tilde{s}$ corresponds to $\hat{s}$.
The transducer $\mathcal{T}$ does not necessarily read its whole input tree.
However, since its input trees are monadic, $\mathcal{T}$ can be modified 
in a similar fashion to Lemma~\ref{only} such that
only right-hand sides of rules for symbols of rank $0$ are ground.
With this modification, it is ensured that $\mathcal{T}$ reads its whole input tree during a translation.

For completeness, we sketch how $T_O$ is obtained from $\mathcal{T}$ and $\mathcal{T}'$, i.e., 
how $\mathcal{T}$ and $\mathcal{T}'$ are run in parallel.
Denote by $\acute{S'}$ and $\hat{q}_0'$ the set of states and the initial state of $\mathcal{T}'$, respectively.
Then the set of states of $T_O$ is $\acute{S}\times \acute{S'}$
while the initial state is $(\hat{q}_0, \hat{q}_0')$.
Consider the symbol $\sigma\in \tilde{\Sigma}_1$.
Let $q( \sigma (x_1))\rightarrow \hat{t}$ be a rule of $\mathcal{T}$ such that  $\hat{q} (x_1)$ occurs in $\hat{t}$.
Let $q'( \sigma (x_1))\rightarrow \hat{t}'$ be a rule of  $\mathcal{T}'$
such that $\hat{q}' (x_1)$ occurs  $\hat{t}'$.
Then $T_O$ contains the rule 
$
(q,q')( \sigma (x_1))\rightarrow \tilde{t}
$
where $\tilde{t}$ is obtained from $\hat{t}$ by replacing occurrences of $\hat{q} (x_1)$
by $(\hat{q}, \hat{q}') (x_1)$. Rules for symbols in  $\tilde{\Sigma}_0$ are defined in the obvious way.

Since 
$A'$ and $N'$ are equivalent on the range of $B$
and due to property that 
on input $s\in \text{range} (B)$ 
an attribute of $A'$ processes the node~$v$ if and only if
a state of $N'$ processes~$v$ on input $s$, 
it follows that $T_{O}$ and $T_W$ are equivalent.
In particular, let $\tilde{s}$ be a tree over $\tilde{\Sigma}$
and let $\tilde{s}$ correspond to $s'\in \text{range} (B)$.
Recall that the domains of $T_W$ and $T_O$ only consist of trees like $\tilde{s}$.
Let $T_W$ produces $t$ on input $\tilde{s}$.
This means that $A'$ produces $t$ on input $s'$ due to Lemma~\ref{lemma 7} which in turn means that
$N'$ produces $t$ on input $s'$
 since 
 $A'$ and $N'$ are equivalent on the range of $B$.
Informally, since the prefix of $s'$ encoded by $\tilde{s}$ is sufficient for $T_W$ to 
simulate $A'$ on input $s'$ and since on input $s'$,
an attribute of $A'$ processes the node~$v$ if and only if
a state of $N'$ processes~$v$ on input $s$, it follows that 
$\tilde{s}$ is sufficient for $T_O$ to 
simulate $N'$ on input $s$. Thus, $T_O$ also produces $t$ on input~$\tilde{s}$. 
 The converse direction follows analogously.
 In summary, the following lemma holds.
 
 \begin{lemma}\label{lemma aux 2}
 	If a $dt^R$ equivalent to the $att^R$ $\hat{A}$
 	exists
 	then a one-way transducer equivalent to $T_W$ exists.
 \end{lemma}
 
 \noindent
Together with  Lemma~\ref{lemma aux}, Lemma~\ref{lemma aux 2} yields the following lemma.

\begin{lemma}\label{iff}
	Let $T_W$ be the two-way transducer obtained from $\hat{A}$.
	A one-way transducer $T_O$ equivalent to $T_W$ exists if and only if
	a $dt^R$ $T$ equivalent to $\hat{A}$ exists.
\end{lemma}

Since it is decidable whether or not 
a one-way transducer equivalent to $T_W$ exists
due to~\cite{DBLP:conf/lics/FiliotGRS13}, we obtain 
the following theorem
with Lemmas~\ref{necessary condition},~\ref{necessary condition decidable},~\ref{lemma aux} and~\ref{iff}. 

\begin{theorem}\label{theorem 1}
	For a $datt$  with monadic output, it is decidable
	whether or not an equivalent $dt^R$   exists and if so then it  can be constructed.
\end{theorem} 

In the following, we will improve the result of Theorem~\ref{theorem 1}.
More precisely, our aim is to show  that even for nondeterministic
 $att^U$ with monadic output
it is decidable
whether or not an equivalent $dt^R$ exists. 
To do so,
we will first show that for a $datt^U$ with monadic output, it is decidable 
whether or not an equivalent $dt^R$ exists.
Afterwards, we will show that 
(a) it is decidable whether or not a nondeterministic $att^U$  with monadic output is functional and 
(b) that every functional $att^U$ with monadic output can be simulated by a  $datt^U$.
This yields the result we aimed for.

\section{Top-Down Definability of Attributed Tree Transducers with Look-Around}
In this section, we show that the result of  Theorem~\ref{theorem 1} can be extended to
$datts$  with look-around and monadic output.

In order to show that for a  $datt^U$ with monadic output, it is decidable 
whether or not an equivalent $dt^R$ exists as well, we first show that the following auxiliary results holds.

\begin{lemma}\label{lemma auiliary}
	Consider the $att^U$ $\breve{A}=(U, A)$. Then an equivalent  $att^U$ $\breve{A}_2=(U_2, A_2)$
	can be constructed such that $\text{dom}(A_2) \subseteq \text{range} (U_2)$.
\end{lemma}

\begin{proof}
First of all, note that since $U$ is a relabeling, its range is effectively recognizable, i.e,
a deterministic bottom-up automaton recognizing it exists and can be constructed.
%
%
%
%
Let $B$ be a deterministic bottom-up tree automaton recognizing the range of $U$.
The idea is as follows. 
Denote by $\Sigma$ the output alphabet of $U$.
Denote by $\Sigma_2$ the  input alphabet of $A_2$.
We define that 
$\Sigma_2$ consists of symbols of the form
$(\sigma, \rho)$ where $\sigma\in \Sigma_k$ and $\rho=\sigma (p_1,\dots,p_k) \rightarrow p(\sigma (x_1,\dots,x_k))$ is a rule of~$B$.
Let $h$ be the tree homomorphism from $T_{\Sigma_2}$ to $T_\Sigma$ given by $h((\sigma, \rho))=\sigma$.
Given the specific form of its input alphabet, $A_2$ tests for an input tree $\breve{s}$ over $\Sigma_2$
whether or not $h(\breve{s})\in \text{range} (U)$.
In the affirmative case, $A_2$ begins to simulate $A$ operating on input $h(\breve{s})$,
otherwise $A_2$ does not produce any output tree.

How does $A_2$ test whether or not $h(\breve{s})\in \text{range} (U)$? 
The idea is that before any output is produced, $A_2$ traverses $\breve{s}$ in  a pre-order fashion.
Let the node $v$ be labeled by $(\sigma, \rho)$, where $\rho= \sigma (p_1,\dots,p_k) \rightarrow p(\sigma (x_1,\dots,x_k))$.
Then for all $i\in [k]$, 
$A_2$ checks whether or not
$v.i$ is labeled by a symbol of the form $(\sigma_i, \rho_i)$ such that $p_i$ occurs on the
right-hand side of $\rho_i$.

We now formally define $A_2=(S_2,I_2, \Sigma_2,\Delta, \tilde{a}_0, R')$.
Let $A=(S,I,\Sigma,\Delta,a_0,R)$.
Denote by  $\text{maxrk}$ the maximal rank of symbols in $\Sigma$.
We define
\[
S_2= S \cup \{\tilde{a}_0, \tilde{a}\} \cup \{a_{\rho, j}\mid \rho \text{ is a rule of } B \text{ and } j\leq \text{maxrk}\}
\]
Furthermore, we define $I_2=I\cup\{\tilde{b},\tilde{b}'\}$.
Informally, the attributes $\tilde{a}$ and $\tilde{b}$ are used to traverse the input tree,
while attributes of the form $a_{\rho, j}$ and $\tilde{b}'$ are used to perform the checks specified above.

The rules of $A_2$ are defined as follows:
Firstly we define $R_\# \subseteq R'_\#$.
Furthermore, we define $\tilde{b}(\pi 1) \rightarrow a_0 (\pi 1) \in R'_{\#}$.

For a symbol $(\sigma, \rho)$ where $\sigma\in \Sigma_k$ and $\rho$ is a rule of the automaton $B$, we define 
that $R_\sigma\subseteq R'_{(\sigma,\rho)}$.

Let $k>0$.
Then we define that $\tilde{a} (\pi)\rightarrow a_{\rho,1} (\pi 1)\in R'_{(\sigma,\rho)}$.
Additionally, if a final state occurs on the right-hand side of $\rho$ then we also define
$\tilde{a}_0 (\pi) \rightarrow a_{\rho,1} (\pi 1) \in R'_{(\sigma,\rho)}$.
Let $\rho'=\sigma' (p'_1,\dots,p'_{k'}) \rightarrow p'(\sigma (x_1,\dots,x_{k'}))$ be a rule of $B$ and let $j\leq k'$.
Then we define 
$a_{\rho',j} (\pi) \rightarrow \tilde{b}' (\pi)\in R'_{(\sigma,\rho)}$ if 
$p_j'$ occurs on the right-hand side of $\rho$.
For the inherited attribute $\tilde{b}'$, we define
$\tilde{b}' (\pi i)\rightarrow a_{\rho, i+1} (\pi (i+1) )\in R'_{(\sigma,\rho)}$ for $i<k$ and
$\tilde{b}' (\pi k)\rightarrow \tilde{a} (\pi 1 )\in R'_{(\sigma,\rho)}$.
Furthermore, for the inherited attribute $\tilde{b}$, we define 
$\tilde{b} (\pi i)\rightarrow \tilde{a} (\pi (i+1) )$ for $i<k$ and
$\tilde{b} (\pi k)\rightarrow \tilde{b} (\pi)$.

Let $k=0$.
Then we define $\tilde{a} (\pi)\rightarrow \tilde{b} (\pi)\in R'_{(\sigma,\rho)}$.
If a final state occurs on the right-hand side of $\rho$ then we additionally define
$\tilde{a}_0 (\pi) \rightarrow \tilde{b} (\pi) \in R'_{(\sigma,\rho)}$. 
For a rule $\rho'=\sigma' (p'_1,\dots,p'_{k'}) \rightarrow p'(\sigma (x_1,\dots,x_{k'}))$ of $B$ and $j\leq k'$,
we define 
$a_{\rho',j} (\pi) \rightarrow \tilde{b}' (\pi)\in R'_{(\sigma,\rho)}$ if 
$p_j'$ occurs on the right-hand side of $\rho$.

Now all that is left is to define $U_2$.
The automaton $B$ induces the deterministic bottom-up relabeling $B'$
as follows:
If $\rho=\sigma (p_1,\dots,p_k) \rightarrow p(\sigma (x_1,\dots,x_k))$ is a rule of  $B$ then 
$\sigma (p_1 (x_1),\dots,p_k (x_k)) \rightarrow p ((\sigma,  \rho) (x_1,\dots,x_k))$
is a rule of  $B'$.
Due to Theorem~2.6 of~\cite{DBLP:journals/mst/Engelfriet77} (and its proof), the composition of a 
$U$ and $B'$  can be simulated by a single  look-around  and this look-around can be constructed.
We define $U_2$  as this look-around.
This concludes the construction of $\breve{A}_2$.
By construction it should be clear that $\breve{A}$ and $\breve{A}_2$ are equivalent.
\end{proof}

Next, let a  $datt^U$ $\breve{A}=(U, A)$ with monadic output,
where $U$ is a top-down relabeling
with look-ahead and $A$ is an $att$ with monadic output,  be given.
Then the following holds.

\begin{lemma}\label{lemma auiliary 2}
	If a $dt^R$ equivalent to $\breve{A}=(U, A)$ exists then
	a $dt^R$ $\hat{T}$ exists such that $A$ and $\hat{T}$ are equivalent on the range of $U$,
	i.e., if $s\in \text{range} (U)$ then $(s,t)\in \tau_{A}$ if and only if
	$(s,t)\in \tau_{\hat{T}}$.
\end{lemma}

\begin{proof}
	Denote by $\Sigma$ and $\Sigma^U$ the input and output alphabet of $U$, respectively.
	W.l.o.g. we can assume 
	that whenever $U$ relabels an input symbol $\sigma\in \Sigma$,  it preserves the original symbol, i.e.,
	we can assume that
	$\Sigma^U$ contains symbols of the form  $\sigma_{\zeta}$ where $\sigma\in \Sigma$ and $\zeta$ is an annotation made to $\sigma$
	and that if a symbol $\sigma\in\Sigma$ is relabeled by $U$ then it is relabeled by a symbol of the form $\sigma_\zeta$.
	
	Denote by $h$ the homomorphism from $T_{\Sigma^U}$ to $T_\Sigma$ defined by
	$h(\sigma_\zeta)=\sigma$.
	Let $\breve{T}=(R,T)$ be a $dt^R$ equivalent  to $\breve{A}$.
	Subsequently, we define the $dt^R$ $\hat{T}=(R',T)$ from 
	$\breve{T}$ such that 
	$A$ and $\hat{T}$ are equivalent on the range of $U$.
	
	Consider the bottom-up relabeling $R$ of $\breve{T}$.
	Obviously the input alphabet of $R$ is $\Sigma$.
	It is easy to see that a bottom-up relabeling $R'$ 
	can be constructed from $R$ such that (a) the input alphabet of $R'$ is $\Sigma^U$
	and (b) $(s,s')\in \tau_{R'}$ if and only if $(h(s),s')\in \tau_R$.
	Informally, $R'$ ignores the annotations of all symbols
	occurring in $s$ and behaves effectively identical to $R$.
	
	To see that $A$ and $\hat{T}$ are equivalent on the range of $U$, we show that if
	$s\in \text{range} (U)$ and $(s,t)\in \tau_{A}$ then $(s,t)\in \tau_{\hat{T}}$.
	The converse follows analogously.
	Since $\breve{T}$ and $\breve{A}$ are equivalent, 	$s\in \text{range} (U)$ and $(s,t)\in \tau_{A}$
	imply that $(h(s), t)\in \tau_{\breve{T}}$ which means that $s'$ exists such that
	$(h(s),s')\in \tau_R$ and $(s',t)\in \tau_T$. By construction of $R'$, $(s,s')\in \tau_{R'}$ holds.
	By definition of $\hat{T}$ this implies $(s,t)\in \tau_{\hat{T}}$.
\end{proof}

Combining Lemmas~\ref{lemma auiliary} and~\ref{lemma auiliary 2}, it follows that 
if a $dt^R$ equivalent to $\breve{A}=(U, A)$ exists then
a $dt^R$ $T$ exists such that $A$ and $T$ are equivalent.
In particular, this follows since by Lemma~\ref{lemma auiliary},
we can assume that  $\text{dom} (A) \subseteq \text{range} (U)$.
Additionally, we can assume that   $\text{dom} (T) \subseteq \text{range} (U)$ holds since
 using its bottom-up relabeling,
$T$ can test whether or not its input tree $s$ is a tree in $\text{range} (U)$
or not. If not then $T$ simply produces no output on input $s$.
Recall that since $U$ is a relabeling, its range is recognizable
and thus $T$ is able to test whether $s\in \text{range} (U)$.

By Theorem~\ref{theorem 1}, it is decidable whether or not a
$dt^R$ $T$ exists such that $A$ and $T$ are equivalent.
If $T$  does not exist then it follows that no $dt^R$ equivalent to $\breve{A}$ exists.
%
On the other hand, the existence of  $T$ 
implies that
a $dt^R$ equivalent to $\breve{A}$ exists. In particular the following holds:
\[
	\tau_{\breve{A}}  = \{(s,t) \mid (s,s') \in \tau_U \text{ and } (s',t) \in \tau_{A}\} 
	 = \{(s,t) \mid (s,s') \in \tau_U \text{ and } (s',t) \in \tau_{T}\} 
\]
By Theorem~2.11 of~\cite{DBLP:journals/mst/Engelfriet77}, $dt^R$ are closed under composition. 
Thus, since $T$ is a $dt^R$ and by definition $U$ is also a $dt^R$,
there exists a $dt^R$ $\hat{T}$ such that $\tau_{\hat{T}} =  \{(s,t) \mid (s,s') \in \tau_U \text{ and } (s',t) \in \tau_{T}\}$,
which yields the following.

\begin{theorem}\label{look-around extension}
	For a  $datt^U$ with monadic output, it is decidable
	whether or not an equivalent $dt^R$   exists and if so then it  can be constructed.
\end{theorem}
\section{Functionality is Decidable for Attributed Tree Transducer with Monadic Output and  Look-Around}
In this section we show that for an $att^U$ $\breve{A}=(U,A)$ with monadic output, it is decidable
whether or not $\breve{A}$ is functional. Note that $\breve{A}$ and hence $A$ may be circular.
Obviously, if $\breve{A}$ is functional, then $A$ must be functional on $\text{range}(U)$, i.e.,
for each $s\in \text{range}(U)$, at most one tree $t$ exists such that $(s, t)\in \tau_{A}$.
Recall that by Lemma~\ref{lemma auiliary}, we can assume that $\text{dom} (A)\subseteq \text{range}(U)$.
Thus, it follows that $A$ itself must be functional if $\breve{A}$ is.
Consequently, it is sufficient to show that it is decidable whether or not $A$ is functional.
The idea is to construct  $datts^R$ $A_1$ and $A_2$ 
such that $A_1$ and $A_2$ are equivalent if and only if
$A$ is functional.
Recall that by Proposition~\ref{equivalent proposition}, equivalence is decidable for $datts^R$.
Hence with $A_1$ and $A_2$ along with Proposition~\ref{equivalent proposition}, 
functionality of $A$ is decidable.  

First consider the following.
Let $A=(S,I,\Sigma,\Delta,a_0,R)$.
Recall that 
$\text{RHS}_A (\#, b(\pi 1))$ denotes the set of all right-hand sides of rules in $R_\#$ that are of the form
$b(\pi 1)\rightarrow \xi$, where $b\in I$.
Recall that the set  $\text{RHS}_A (\sigma, a(\pi))$ where $a\in S$ and $\sigma\in \Sigma$ is defined analogously. 

\begin{lemma}\label{assumption}
	Let $A=(S,I,\Sigma,\Delta,a_0,R)$ be an $att$.
	Then an equivalent $att$ $A'=(S',I,\Sigma,\Delta,a_0,R')$  can be constructed such that
	$R'_\#$  contains no distinct rules with the same left-hand side and such that
	$A'$ is only circular if $A$ is. 	
\end{lemma}
\begin{proof}
	In the following denote by $\mathcal{I}$ the set $\{b\in I \mid |\text{RHS}_A (\#, b(\pi 1))|>1 \}$.
	We define $S'=S\cup \{a_b \mid b\in \mathcal{I}\}$ and
	\[
		R'_\#  = \{ b(\pi 1) \rightarrow \xi \mid  b\in  I\setminus \mathcal{I},\  \xi \in  \text{RHS}_A (\#, b(\pi 1)) \}\
		\cup \ \{ b(\pi 1) \rightarrow a_b  (\pi 1)\mid b\in \mathcal{I} \}.
	\]
For $\sigma\in \Sigma$ we define
	\[
	\begin{array}{ccl}
		R_\sigma' & = &R_\sigma\ \cup \ \{a_b \rightarrow \xi \mid b\in \mathcal{I} \text{ and } \xi \in T_\Delta\cap \text{RHS}_A (\#, b(\pi 1)) \}\\
		& \cup & \{a_b \rightarrow \zeta \mid b\in \mathcal{I} \text{ and }  \exists a\in S,\ \psi\in  \text{RHS}_A (\sigma, a(\pi )),\  \xi \in  \text{RHS}_A (\#, b(\pi 1)): \\
		&&  
		a(\pi 1) \text{ occurs in }\xi \text{ and }
		\zeta=\xi [v\leftarrow \psi \mid v\in V(\xi),\ \xi [v] =a(\pi 1)]\}.\\
	\end{array}
	\]
	It can be shown by straight-forward structural induction that $A$ and $A'$ are equivalent.
	By construction, it follows that if $A$ is noncircular then $A'$ is too.  
\end{proof}

Recall that we aim to construct  $datts^R$ $A_1$ and $A_2$ 
such that $A_1$ and $A_2$ are equivalent if and only if
$A$ is functional.
For simplicity and ease of understanding, we first consider the case where $A$ is noncircular.
We will later show how to generalize our procedure to the case that $A$ is circular.
The idea for the construction of $A_1$ and $A_2$ is similar to the one in Lemma~2.9 of~\cite{schmudethesis}:
the input alphabet encodes which rules are allowed to be applied.
The input alphabet $\hat{\Sigma}$ of $A_1$ and $A_2$ contains symbols of the form
\[
\begin{array}{ll}
	\langle\sigma, R^1,R^2\rangle,&
	\text{where }\sigma\in \Sigma \text{ and } 
	R^1,R^2\subseteq R_\sigma \text{ such that }  \text{for }i\in [2],\\
	&
	\text{no rules in }R^i \text{ have the same left-hand side.}
\end{array}
\]
The symbol $\langle\sigma, R^1,R^2\rangle$ 
has the same rank as $\sigma$.
Informally, the idea is that the label $\langle\sigma, R^1,R^2\rangle$ of a node $v$ 
determines which rules 
  $A_1$ and $A_2$ may apply at $v$.
In particular, $A_1$ is only allowed to apply rules in $R^1$.
Likewise, $A_2$ is restricted to applying rules in~$R^2$.
Note that due to Lemma~\ref{assumption} the rules of $A$ for the root marker can be assumed to be deterministic.
Therefore and since by definition no rules in  $R^i$ have the same left-hand side for $i\in [2]$,
it should be clear that $A_1$ and $A_2$ are both deterministic.

More formally, we define $A_1=(R,A_1')$ and $A_2=(R,A_2')$.
The look-ahead $R$ checks whether or not an input tree $s$ is in $\text{dom} (A_1') \cap \text{dom} (A_2')$.
If not then neither $A_1$ nor $A_2$ produces an output tree on input $s$.
Hence, the domain of $A_1$ and $A_2$ is $\text{dom} (A_1') \cap \text{dom} (A_2')$.
Note that by~\cite{DBLP:journals/ipl/FulopM00},  $\text{dom} (A_1')$ and $\text{dom} (A_2')$ are recognizable.

We define $A_1'=(S,I,\hat{\Sigma},\Delta,a_0,\hat{R})$.
The rules in $\hat{R}$ are defined as follows:
For a symbol of the form $\breve{\sigma}=	\langle\sigma, R^1,R^2\rangle$
we define that
if $\rho\in R^1$  then $\rho \in \hat{R}_{\breve{\sigma}}$.
The rules for the root marker, we define  
that $\hat{R}_{\#}= R_\#$.
Recall that  due to Lemma~\ref{assumption} the rules of $A$ for the root marker can be assumed to be deterministic.
This concludes the construction of $A_1'$. The $att$ $A_2'$ is constructed analogously.

\begin{lemma}\label{eq func}
	The $atts^R$ $A_1$ and $A_2$ are equivalent if and only if $A$ is functional.
\end{lemma}

\begin{proof}
	Denote by $h$ the homomorphism from $T_{\hat{\Sigma}}$ to $T_\Sigma$
	defined by $h(\langle\sigma, R^1,R^2 \rangle )=\sigma$.
	By definition $A_1$ and $A_2$ have the same domain.
	Assume that $A_1$ and $A_2$ are not equivalent. 
	Hence, $\hat{s} \in T_{\hat{\Sigma}}$ and $t_1,t_2\in T_\Delta$
	exists such that $t_1\neq t_2$ and $(\hat{s}, t_i)\in \tau_{A_i}$ for $i\in [2]$.
	By construction of $A_1$ and $A_2$,
	the latter obviously implies that  $( h(\hat{s}), t_i)\in \tau_{A}$ for $i\in [2]$.
	Thus, $A$ is not functional.
	
	Before we prove the converse, consider the following.
	Let $s\in T_\Sigma$ and $t\in T_\Delta$ such that $(s,t)\in \tau_A$.
	In particular, let $t_1,\dots,t_n$ be trees such that
	\[
	\tau=(a_0 (1)= t_1 \Rightarrow_{A,s^\#} \cdots \Rightarrow_{A,s^\#} t_{n} =t).
	\]
	Then trees $s_1\in T_{\hat{\Sigma}}$ exists such
	that $h(s_1)=s$ and $(s_1,t)\in \tau_{A_1'}$. In particular the trees 
	$s_1$ are of the following form:
	Denote by $\tau [v]$ the set of all rules applied at the node $v$ in
	$\tau$.
	More formally, let $i<n$, $a\in S$ and $b\in I$.
	If $a(v)$ occurs in $t_i$ and
	$t_i  \Rightarrow_{A, s^\#} t_{i+1}$ is due to the rule
	$a(\pi)\rightarrow \xi$ then $a(\pi)\rightarrow \xi$ is contained in $\tau[ v ]$.
	If $b(v.j)$ occurs in $t_i$ and
	$t_i  \Rightarrow_{A, s^\#} t_{i+1}$ is due to the rule
	$b(\pi j) \rightarrow \xi'$  then $b(\pi j) \rightarrow \xi'$ is contained $ \tau [ v ]$.
	Note that since $A$ is noncircular, no distinct rules with the same left-hand-side occur in $\tau [ v]$.
	Let $s_1$ be such that if the node $v$ is labeled by the symbol $\sigma$ in $s^\#$
	then $v$ is labeled by a node of the form $\langle\sigma, \tau[v], R \rangle$ where $R$ is an arbitrary subset of $R_\sigma$.
	Then clearly, 
		$
	t_1 \Rightarrow_{A_1',s_1^\#} \cdots \Rightarrow_{A_1',s_1^\#} t_{n},
	$
	which yields our claim. Analogously, it can be shown that trees $s_2\in T_{\hat{\Sigma}}$ exists such
	that $h(s_2)=s$ and $(s_2,t)\in \tau_{A_2'}$.
	Specifically, the trees $s_2$ are of the following form:
	If the node $v$ is labeled by the symbol $\sigma$ in $s^\#$
	then $v$ is labeled by a node of the form $\langle\sigma, R,\tau[v],\rangle$ where $R$ is an arbitrary subset of $R_\sigma$.
	
	We now show that if $A$ is not functional then $A_1$ and $A_2$ are not equivalent.
	Let $s\in T_\Sigma$ and $t_1,t_2\in T_\Delta$ exist such that 
	$t_1\neq t_2$ and $(s, t_i)\in \tau_{A}$ for $i\in [2]$. 
	Consider the corresponding translations
	\[
	\tau_i=(t_1^i \Rightarrow_{A,s^\#} \cdots \Rightarrow_{A,s^\#} t_{n_i}^i),
	\]
	where $t_1^i =a_0 (1)$ and $t_{n_i^i}=t_i$.
	Define the tree $\hat{s}\in T_{\hat{\Sigma}}$ such that
	if the node $v$ is labeled by $\sigma$ in $s^\#$ then $v$ is labeled by $\langle\sigma, \tau_1 [ v],  \tau_2 [v]\rangle$ in $\hat{s}^\#$.
	Due to previous considerations, for $i=1,2$,
	$t_1^i \Rightarrow_{A_i',\hat{s}^\#} \cdots \Rightarrow_{A_i',\hat{s}^\#} t_{n_i}^i$.
	 Hence, $\hat{s}\in \text{dom} (A_1')\cap \text{dom} (A_2')$.
	 Altogether this implies that $A_1$ and $A_2$ are not equivalent.
\end{proof}

\noindent
With Lemma~\ref{eq func} and Proposition~\ref{equivalent proposition} the following holds.

\begin{lemma}\label{noncircular}
	For a noncircular $att$ $A$, it is decidable whether or not $A$ is functional.
\end{lemma}

Subsequently, we discuss the case where $A$ is circular.
First consider the following definition.
Let $s\in \text{dom} (A)$ and let $\tau$ be a translation of $A$ on input $s$.
In particular, let
 $\tau$ be 
\[
a_0 (1) = t_1 \Rightarrow_{A, s^\#} t_2 \Rightarrow_{A, s^\#}  \cdots \Rightarrow_{A, s^\#} t_n \in T_\Delta
\]
where $t_2,\dots, t_{n-1}\in T_\Delta [\text{SI} (s^\#)]$.
Since $A$ is circular, $\alpha (\nu) \in \text{SI}(s^\#)$ and distinct $i,j\in [n]$ may
exist such that $\alpha (\nu)$ occurs in $t_i$ and $t_j$.
In this case we say that $\tau$ contains a \emph{cycle}. 
Specifically, $\tau$ contains a \emph{productive cycle} if $t_i\neq t_j$.
We say that $\tau$ is  \emph{cycle-free} if $\tau$ contains no cycle.
It is easy to see that $A$ cannot be functional
if $\tau$ contains a productive cycle.

In the following, we show that it is decidable whether or not a translation of $A$
containing a productive cycle
exists. To do so consider the following the following observation.

	\begin{observation}\label{productive cycle}
	Let $A=(S,I,\Sigma,\Delta, a_0, R)$.
	Assume that a translation of $A$ on input $s\in T_\Sigma$ containing a productive cycle
	exists. Then, in particular a translation
	\[
	a_0 (1) \Rightarrow_{A, s^\#} t_1 \Rightarrow_{A, s^\#}  \cdots \Rightarrow_{A, s^\#} t_n  \Rightarrow_{A, s^\#} t\in T_\Delta,
	\]
	and $i<j\leq n$ exist such that $t_i\neq t_j$
	 and for some $a\in S$ and for some node $v$ it holds that
	$a(v)$ occurs in $t_i$ and $t_j$.
\end{observation}

\noindent
It should be clear that Observation~\ref{productive cycle} holds. 
Recall that due to Proposition~\ref{empty test}, it is decidable whether or not the range of an $att^R$ 
intersected with a recognizable tree language is empty.
Hence with Observation~\ref{productive cycle}, the following holds.

\begin{lemma}\label{productive cycle decide}
	Let $A=(S,I,\Sigma,\Delta, a_0, R)$ be an $att$. It is decidable whether or not a translation $\tau$
	 of $A$ containing a productive cycle
	exists. 
\end{lemma}

\begin{proof}
	To decide whether such a translation $\tau$ exists, we construct an $att^R$
	$\bar{A}=(R,A')$ from~$A$. The idea is as follows: 
	Let $s\in T_\Sigma$.
	In the following, nodes of $s$ may be marked by having their labels annotated by $\pm$.
	We demand that $\bar{A}$ only produces output for input trees where precisely one node is marked.
	Whether or not an input tree has exactly one marked node is tested by the look-ahead $R$.
	The $att$ $A'$ is constructed such that
	whenever a marked node $v$ is processed by a synthesized attribute $a$, we output $a$ as well.	
	
	More formally, $A'=(S,I,\Sigma',\Delta',a_0,R')$
	where $\Sigma'=\Sigma \cup \{\sigma_\pm \mid \sigma \in \Sigma\}$ and  $\sigma_\pm$ is of rank $k$ if $\sigma$ is.
	We define $\Delta'= \Delta\cup S$ where elements in $S$ are considered to be of rank $1$.
	The rules of $A'$ are defined as follows:
	We define  $R'_\#=R_\#$ and $R'_\sigma =R_\sigma$ for $\sigma \in \Sigma$.
	Consider a symbol of the form $\sigma_\pm$. 
	If $b(\pi j)\rightarrow \xi \in R_\sigma$ where $b\in I$, then  $b(\pi j)\rightarrow \xi \in R_{\sigma_\pm}'$.
	If $a(\pi)\rightarrow \zeta \in R_\sigma$ where $a\in S$, then  $a(\pi)\rightarrow a (\zeta) \in R_{\sigma_\pm}'$.
	
	Consider a tree $t\in T_{\Delta'}$ for which
	nodes
	$u_1,u_2,u_3 \in V(t)$ exist such that  $u_i$ is an ancestor of $u_{i+1}$ for $i<3$ and
	$t[u_1] = t[u_3] \in S$ while $t[u_2]\in \Delta$.
	The set $L$ containing all such trees is regular.
	Due to Observation~\ref{productive cycle} it is easy to see that 
	translation $\tau$
	of $A$ containing a productive cycle
	exists if and only if $\text{range} (\bar{A})\cap L \neq \emptyset$.
	By Lemma~\ref{empty test}, the latter is decidable.
\end{proof}

Since productive cycles cause nonfunctionality and by
Lemma~\ref{productive cycle decide}, it is decidable whether a
translation $\tau$
of $A$ containing a productive cycle
exists, we subsequently assume that $A$ is \emph{productive cycle-free}, meaning that no translation of
$A$ contains a productive cycle.
Obviously, this means that translations of $A$ may still contain nonproductive cycles, however these are easy to deal with.
In particular, if $A$ is productive cycle-free then 
we can decide whether or not $A$ is functional using the same procedure as in the case where $A$ is noncircular, i.e.,
we construct $atts^R$ $A_1$ and $A_2$ such that $A_1$ and $A_2$ are equivalent if and only if $A$ is functional.
In particular $A_1$ and $A_2$ are constructed as in the case where $A$ is noncircular.
First consider the following observation.

\begin{observation}\label{cycle free}
    Let $A$ be a productive cycle-free $att$.
    Let $s\in T_\Sigma$ and let $t\in T_\Delta$.
    Consider a translation $\tau$ which yields $t$ on input $s$.
    Then a cycle-free translation
    $\tau'$ which yields $t$ on input $s$ exists.
\end{observation}

\noindent
It is easy to see that Observation~\ref{cycle free} holds.
With Observation~\ref{cycle free},
the following result holds.

\begin{lemma}\label{eq circ}
Let $A$ be a productive cycle-free $att$.
The $atts$ $A_1$ and $A_2$ are equivalent if and only if $A$ is functional.
\end{lemma}

\begin{proof}
The if-direction follows as in Lemma~\ref{eq func}.
For the only-if direction,  
let $s\in T_\Sigma$ and $t_1,t_2\in T_\Delta$  such that 
$t_1\neq t_2$ and $(s, t_i)\in \tau_{A}$ for $i\in [2]$. 
Consider the corresponding translations
\[
\tau_i=(a_0 (1) \Rightarrow_{A,s^\#} t_1^i \Rightarrow_{A,s^\#} \cdots \Rightarrow_{A,s^\#} t_{n_i}^i\Rightarrow_{A,s^\#} t_i).
\]
By Observation~\ref{cycle free}, $\tau_1$ and $\tau_2$ can be assumed to be cycle-free.
For a node $v\in V(s^\#)$, define the sets $\tau_1[v]$ and $\tau_2[v]$ 
as  the set of all rules applied at the node $v$ in
$\tau_1$ and $\tau_2$, respectively,
as in Lemma~\ref{eq func}.
Note that since $\tau_1$ and $\tau_2$ are cycle-free, 
no distinct rules with the same left-hand-side occur in $\tau_1 [v]$ and $\tau_2 [v]$.
Given the  sets $\tau_1 [v]$ and $\tau_2 [v]$ for each $v\in V(s^\#)$, we construct a tree $\hat{s}$
such that on input $\hat{s}$, $A_1$ outputs $t_i$ as in Lemma~\ref{eq func}.
This yields the only-if direction.
\end{proof}

\noindent
Due to Lemmas~\ref{productive cycle decide} and~\ref{eq circ}, the following holds.

\begin{lemma}\label{circular}
	For a circular $att$ $A$, it is decidable whether or not $A$ is functional.
\end{lemma}
With the considerations at the start of the section,
Lemma~\ref{circular} yields the following.

\begin{theorem}\label{functional theorem 1}
	It is decidable whether an $att^U$ $\breve{A}=(U,A)$ is functional. 
\end{theorem}

\section{From Functional Attributed Tree Transducers to Deterministic Attributed Tree Transducers}\label{functional subsection}
Denote by $ATT^U$ and $ATT^U_\text{mon}$, the classes of tree translations realizable
by $atts^U$  and  $atts^U$ with monadic output, respectively.
Analogously, denote by $dATT^U$ and $dATT^U_\text{mon}$
the classes of tree translations realized by deterministic such transducers.

Subsequently, we show that
$ATT^U_{\text{mon}} \cap \text{func} = dATT^U_\text{mon}$,
where $\text{func}$ denotes the class of all functions.
First, consider the following result which holds
due to Theorem~35 of~\cite{DBLP:journals/acta/EngelfrietIM21}\footnote{
Note that $atts^U$
are  ($TT$s) as defined~\cite{DBLP:journals/acta/EngelfrietIM21}.
Note that in~\cite{DBLP:journals/acta/EngelfrietIM21}, $dTT_\downarrow$ denotes a deterministic top-down transducer with look-around.
For deterministic top-down transducers, look-around is the same as
look-ahead since $dt^R$ are closed under composition~\cite{DBLP:journals/mst/Engelfriet77}.
See also Lemma~12 in~\cite{DBLP:journals/acta/EngelfrietIM21}. 
Thus, a $dTT_\downarrow$ is basically $dt^R$.
Note that a $dTT_\downarrow$  can be assumed to be `stay-free', i.e., it does not have stay-rules.}.
Note that $\circ$ denotes the composition of two classes of binary relations.

\begin{proposition}\label{prop 1}
$ATT^U \cap \text{func} \subseteq DT^R \circ dATT^U$.
\end{proposition}

Recall that for a tree $s$, the \emph{size} of $s$ is $|s|:= |V(s)|$.
A function $\tau: T_\Sigma \rightarrow T_\Delta$ is of \emph{linear size} increase if a constant $c\in \mathbb{N}$ exists such that
$|\tau (s)| \leq c\cdot |s|$. Denote by $\text{LSIF}$, the class of all functions of linear size increase.
Theorem~43 of~\cite{DBLP:journals/acta/EngelfrietIM21} implies the following result.

\begin{proposition}\label{prop 2}
	$(dATT^U \circ dATT^U )\cap \text{LSIF} =  dATT^U$.
\end{proposition}

By~\cite{DBLP:journals/acta/EngelfrietIM21}, Propositions~\ref{prop 1} and~\ref{prop 2} are effective.
For a functional $att^U$ with monadic output, we show that the following holds.

\begin{proposition}\label{linear size lemma}
	Any functional $att^U$ $A$ with monadic output is of linear sice increase.
\end{proposition}
\begin{proof}
	Our  proof is analogous to the one for Proposition~\ref{datt lsi}.
	Let $\breve{A}=(U,A)$ be a  functional $att^U$ with monadic output.
	Due to Lemma~\ref{lemma auiliary}, we can assume that $A$ is a functional $att$.
	To show that  $\breve{A}$ is of linear size increase,
	it is clearly sufficient to show that $A$ is of linear size increase.
	Let $(s,t)\in \tau_A$.  Since $A$ is functional, trees $t_1,\dots,t_n \in T_\Delta [SI(s^\#)]$ exist such that
	$
	a_0 (1) =  t_1\Rightarrow_{A, s^\#} \cdots \Rightarrow_{A, s^\#} t_n \Rightarrow_{A, s^\#} t
	$
	is cycle-free.
	Thus, for all $\alpha (\nu) \in \text{SI}(s^\#)$ at most one $j\in [n]$ exists
	such that $\alpha (\nu)$ occurs in $t_j$. 
	Analogously as in  Proposition~\ref{datt lsi} it follows that $|t| \leq  \text{maxsize}\cdot |S\cup I| \cdot |s|$,
	where $\text{maxsize}$ denotes the maximal size of a right-hand side of a rule of $A$.
\end{proof}

\noindent
Since Propositions~\ref{prop 1} and \ref{prop 2} are effective, Proposition~\ref{linear size lemma} yields
the following result:

\begin{theorem}\label{functional theorem 2}
	For any functional $att^U$ with monadic output an equivalent deterministic $att^U$ can be constructed.
\end{theorem}
\begin{proof}
Let $\breve{A}$ be a functional $att^U$ with monadic output.
By Proposition~\ref{prop 1}, $\breve{A}$ is equivalent to the composition of a $dt^r$ $T$ and  is a  $datt^U$ $D$.
Since, $\breve{A}$ is of linear size increase (Lemma~\ref{linear size lemma}), so is the composition of $T$ and~$D$.
This means that due to Proposition~\ref{prop 2}, a  $datt^U$ equivalent to the composition of $T$ and $D$
exists and can be constructed.
\end{proof}
Theorem~\ref{functional theorem 2} yields the following corollary. 
 
 \begin{corollary}
 	$ATT^U_{\text{mon}} \cap \text{func} = dATT^U_\text{mon}$.
 \end{corollary}
\section{Final Results}
We have shown in Theorem~\ref{functional theorem 1} that for any $att^U$ with monadic output, functionality is decidable.
Furthermore, note that by Theorem~\ref{functional theorem 2}, 
for any functional $att^U$ with monadic output,
an equivalent deterministic $att^U$ can be constructed.
Finally, by Theorem~\ref{look-around extension}, it is decidable for a deterministic $att^U$  with monadic output,
whether or not an equivalent $dt^R$ exists.
Combining these results, we obtain  the following theorem.

\begin{theorem}
	For any $att^U$ with monadic output, it is decidable
	whether or not an equivalent $dt^R$   exists and if so then it  can be constructed.
\end{theorem}

We remark that any  $att^U$ with monadic output  equivalent to some $dt^R$ must obviously be functional. 
Note that by definition $dt^R$s with monadic output are  \emph{linear}. For linear
$dt^R$s it is decidable whether or not an equivalent linear $dt$ exists~\cite{DBLP:conf/icalp/ManethS20} and
if so then such a $dt$ can be constructed. Hence, the following corollary holds.

\begin{corollary}
	For a any $att^U$  with monadic output, it is decidable
whether or not an equivalent $dt$   exists and if so then it  can be constructed.
\end{corollary}
\section{Conclusion}
We have shown how to decide for a given (circular, partial, nondeterministic) attributed transducer with look-around but
restricted to monadic output,  whether or not
an equivalent deterministic top-down
tree transducers (with or without look-ahead) exists
and whether or not 
it is functional. 
Clearly we would like to extend the definability result to non-monadic output trees, i.e.,
we would like to show how to decide for a given arbitrary attributed tree transducer whether or not an equivalent deterministic top-down
tree transducers (with or without look-ahead) exists.
The latter seems quite challenging, as it is not clear whether or not the
result~\cite{DBLP:journals/lmcs/BaschenisGMP18} can be applied in this case.
A decision procedure for the functionality of arbitrary attributed tree transducer implies that 
equivalence of attributed tree transducer is decidable.
The latter  is  a long standing open problem.

\bibliographystyle{splncs04}
\bibliography{mybib}

\end{document}